\documentclass[11pt]{article}

\usepackage{amssymb,latexsym,amsmath,amscd,slashed,mathtools,mathrsfs,bm,stmaryrd,dsfont,amsthm}
\usepackage{setspace,scalerel,color,graphicx,enumerate,hyperref}
\usepackage[multiple]{footmisc}
\usepackage{authblk}

\usepackage{geometry}
\allowdisplaybreaks

\usepackage{tocloft}

\usepackage{sectsty}
\sectionfont{\fontsize{12}{15}\selectfont}

\newtheorem{theorem}{Theorem}[section]
\newtheorem{lemma}[theorem]{Lemma}
\newtheorem{proposition}[theorem]{Proposition}

\newtheorem{definition}{Definition}[section]
\newtheorem{hypothesis}{Hypothesis}[section]

\numberwithin{equation}{section}

\newcommand{\di}{\mathop{}\!\mathrm{d}}
\newcommand{\ci}{\mathrm{i}}
\newcommand{\tr}{\mathrm{Tr}}
\newcommand{\les}{\leqslant}
\newcommand{\gre}{\geqslant}
\newcommand{\rt}{\mathrm{t}}
\newcommand{\ft}{\bm{\mathrm{t}}}
\newcommand{\emm}{\hat{\mu}_N}
\newcommand{\res}{\mathrm{Res}}

\title{\Large{\textbf{Random Finite Noncommutative Geometries\\
and\\
Topological Recursion}}}

\author{Shahab Azarfar and Masoud Khalkhali}

\affil{Department of Mathematics, University of Western Ontario\\
London, Ontario, Canada\footnote{\emph{Email addresses}: sazarfar@uwo.ca, masoud@uwo.ca}}

\date{}

\begin{document}

\maketitle

\begin{abstract}
In this paper we investigate a model for quantum gravity on finite noncommutative spaces using the theory of blobbed topological recursion. The model is based on a particular class of random finite real spectral triples ${(\mathcal{A}, \mathcal{H}, D , \gamma , J) \,}$, called random matrix geometries of type ${(1,0) \,}$, with a fixed fermion space ${(\mathcal{A}, \mathcal{H}, \gamma , J) \,}$, and a distribution of the form ${e^{- \mathcal{S} (D)} {\mathop{}\!\mathrm{d}} D}$ over the moduli space of Dirac operators. The action functional ${\mathcal{S} (D)}$ is considered to be a sum of terms of the form ${\prod_{i=1}^s \mathrm{Tr} \left( {D^{n_i}} \right)}$ for arbitrary ${s \geqslant 1 \,}$. The Schwinger-Dyson equations satisfied by the connected correlators ${W_n}$ of the corresponding multi-trace formal 1-Hermitian matrix model are derived by a differential geometric approach. It is shown that the coefficients ${W_{g,n}}$ of the large $N$ expansion of ${W_n}$'s enumerate discrete surfaces, called stuffed maps, whose building blocks are of particular topologies. The spectral curve 
${\left( {\Sigma , \omega_{0,1} , \omega_{0,2}} \right)}$ of the model is investigated in detail. In particular, we derive an explicit expression for the fundamental symmetric bidifferential ${\omega_{0,2}}$ in terms of the formal parameters of the model.
\end{abstract}

\tableofcontents

\section{Introduction and basics}

In metric noncommutative geometry, the formalism of \emph{spectral triples} \cite{Connes-1995} encodes, in the commutative case, the data of a  Riemannian metric  on a spin manifold in terms of the Dirac operator. More precisely, by \emph{Connes' reconstruction theorem} \cite{Connes-2013} one knows that a spin Riemannian manifold can be fully constructed if we are given a commutative spectral triple satisfying some natural conditions like reality, and regularity. A simple manifestation of this fact is a distance formula \cite{Connes-1989} according to which one can recover the Riemannian metric from the interaction between the Dirac operator and the algebra of smooth functions on a spin manifold through their actions on the Hilbert space of ${L^2}$-spinors. This naturally leads to the view that spectral triples in general, without commutativity assumption, can be regarded as noncommutative spin Riemannian manifolds. 

The data of a real spectral triple ${(\mathcal{A}, \mathcal{H}, D , \gamma , J)}$ consists of a ${^\ast}$-algebra ${\mathcal{A}}$ together with a ${^\ast}$-representation ${\pi : \mathcal{A} \to \mathcal{L} (\mathcal{H})}$ on a Hilbert space $\mathcal{H}$, a self-adjoint Dirac operator $D$, a ${\mathbb{Z}/2}$-grading $\gamma$, and an anti-linear isometry $J$ acting on $\mathcal{H}$. The above-mentioned operators should satisfy certain (anti)-commutation relations and technical functional analytic conditions (see \cite{Connes-1995}, \cite{Connes-Chamseddine-Marcolli-2007} and \cite{Connes-Marcolli-NCG-QFT} for the detailed axiomatic definition of a real spectral triple).

A spectral triple ${\left( { \mathcal{A}, \mathcal{H} , D} \right)}$ is called \emph{finite} if the Hilbert space $\mathcal{H}$ is finite dimensional, i.e. ${\mathcal{H} \cong \mathbb{C}^n \, }$. The data 
${(\mathcal{A} , \mathcal{H} , \gamma , J)}$ corresponding to a finite real spectral triple is referred to as a \emph{fermion space}. Given a fixed fermion space ${(\mathcal{A} , \mathcal{H} , \gamma , J)}$, the \emph{moduli space} of Dirac operators of the finite real spectral triple 
${(\mathcal{A}, \mathcal{H}, D , \gamma , J)}$ consists of all possible self-adjoint operators $D$ (up to unitary equivalence) which satisfy the axiomatic definition of a real spectral triple \cite{Marcolli-NoncomCosmology}. It is considered as the space of all possible geometries, that is Riemannian metrics, over the noncommutative space ${(\mathcal{A} , \mathcal{H} , \gamma , J)}$.

The theory of spectral triples has been used in constructing geometric models of matter coupled to gravity, using an action functional, called the \emph{spectral action}, which is given in terms of the spectrum of the  Dirac operator (see \cite{Connes-1996}, \cite{Connes-Chamseddine-1997}, \cite{Connes-2006} and \cite{Connes-Chamseddine-Marcolli-2007}; see also \cite{Gesteau-2019} for a recent work in the spirit of matrix models). 

This paper is about a second application of the idea of spectral triples by creating a connection with the recently emerged theory of \emph{topological recursion} \cite{Eynard-Orantin-2007}. 
Roughly speaking, if we understand quantization of gravity as a path integral over the space of metrics, it is natural to consider models of Euclidean quantum gravity over a finite noncommutative space in which one integrates over the moduli space of Dirac operators for a fixed fermion space. Given a fermion space ${(\mathcal{A} , \mathcal{H} , \gamma , J)}$, the distribution over the moduli space of Dirac operators is considered to be of the form
\begin{equation}
e^{- \mathcal{S} (D)} \di D \, ,
\end{equation}
where the action functional ${\mathcal{S} (D)}$ is defined in terms of the spectrum of the Dirac operator $D$.

The investigation of the relation between models of quantum gravity on a certain class of finite noncommutative spaces and (anti)-Hermitian matrix ensembles started in the work of Barrett and Glaser (\cite{Barrett-Glaser-2016}, cf. also \cite{Barrett-2015}), although largely through numerical 
simulations. In this paper we consider a much larger class of models and show that an analytic approach to analyzing these models is possible, using techniques of \emph{topological recursion} 
and \emph{blobbed topological recursion} pioneered by Eynard, Orantin 
\cite{Eynard-Orantin-2007, Eynard-Orantin-2009}, Chekhov \cite{Chekhov-Eynard-Orantin-2006}, and 
Borot \cite{Borot-2014, Borot-Eynard-Orantin-2015, Borot-Shadrin-2017}.  

In the following, we recall the definition of a particular type of finite real spectral triples whose Dirac operators are classified in terms of (anti)-Hermitian matrices in \cite{Barrett-2015}. Denote the real Clifford algebra associated to the vector space ${\mathbb{R}^n}$ and the pseudo-Euclidean metric ${\eta}$ of signature ${(p,q) \,}$, given by 
\begin{equation} 
{\eta} (v,v) = {{v_1}^2} + \cdots + {{v_p}^2} - {{v_{p+1}}^2} - \cdots -  {{v_{p+q}}^2} \,
, \quad v \in \mathbb{R}^n \, ,
\end{equation}
by ${{\mathrm{C \ell}}_{p,q} \,}$.\footnote{Given a quadratic form ${\mathrm{q}}$ on a vector space $V$, we follow the convention to define the Clifford algebra ${\mathrm{C \ell} (V , \mathrm{q})}$ associated to $V$ and $\mathrm{q}$ as ${ {\mathrm{C \ell} (V , \mathrm{q})} \coloneqq \mathcal{T} (V) / \mathcal{I}_{\mathrm{q}} (V)\, }$, where ${\mathcal{T} (V) = \sum_{r=0}^{\infty}  V^{\otimes r} \,}$, and ${\mathcal{I}_{\mathrm{q}} (V)}$ denotes the two-sided ideal generated by elements of the form ${v \otimes v - \mathrm{q}(v) \mathds{1}}$ for ${v \in V}$.} Consider the 
complexification ${ {\mathbb{C} \ell}_n \coloneqq {\mathrm{C \ell}}_{p,q} {\otimes_{\mathbb{R}}} \mathbb{C}\,}$ of ${{\mathrm{C \ell}}_{p,q} \,}$. Let ${ {\{ e_i \}}_{i=1}^n }$ be the standard oriented orthonormal basis, i.e. ${{\eta} (e_i , e_j) = \pm \delta_{ij} \,}$, for ${\mathbb{R}^n}$. The \emph{chirality operator} $\Gamma$ is given by
\begin{equation}
\Gamma = \ci^{\frac{1}{2} s (s+1)} \, e_1 e_2 \cdots e_n \, ,
\end{equation}
where ${s \equiv q-p \pmod{8} \, , \ 0 \les s < 8 \,}$. 
We denote by ${V_{p,q}}$ the unique (up to unitary equivalence) irreducible complex ${{\mathrm{C \ell}}_{p,q} \,}$-module, where, for $n=p+q$ odd, the chirality operator $\Gamma$ acts trivially on ${V_{p,q} \,}$.\footnote{Let ${\rho : {\mathrm{C \ell}}_{p,q} \to \text{Hom}_{\mathbb{C}} (V , V)}$ be an irreducible complex unitary representation of ${{\mathrm{C \ell}}_{p,q} \,}$. It can be 
shown \cite{Lawson-SpinGeo} that: 
\begin{itemize}
\item
If ${n=p+q}$ is even, then the representation $\rho$ is unique up to unitary equivalence;
\item
If ${n=p+q}$ is odd, then either ${\rho (\Gamma) = \mathds{1}_V}$ or 
${\rho (\Gamma) = - \mathds{1}_V}$. Both possibilities can occur, and the corresponding representations are inequivalent.
\end{itemize}}\footnote{We have ${V_{p,q} \cong \mathbb{C}^k}$, where ${k= 2^{n/2}}$ (resp. ${k= 2 ^{(n-1)/2}}$) for $n$ even (resp. odd).} We refer to ${\gamma^i = \rho (e_i) \, , \,}$ ${i=1 , \cdots , n \,}$, i.e. the Clifford multiplication by ${e_i}$'s, as the \emph{gamma matrices}. There exist a Hermitian inner product ${\langle \cdot , \cdot \rangle}$ on ${V_{p,q}}$ such that the gamma matrices act unitarily with respect to it, i.e. ${\langle \gamma^i u , \gamma^i v \rangle = \langle u , v \rangle \, , \,}$ ${i=1 , \cdots , n}$ \cite{Lawson-SpinGeo}.\footnote{For ${i=1,\cdots , p}$ (resp. ${i= p+1 , \cdots , n}$) the gamma matrix ${\gamma^i}$ is Hermitian (resp. anti-Hermitian) with respect to 
${{\langle \cdot , \cdot \rangle} \,}$.} 

Consider the Hilbert space ${\left( { V_{p,q} \, , \, {\langle \cdot , \cdot \rangle}} \right)}$. Let ${C : V_{p,q} \to V_{p,q}}$ be a \emph{real structure} of ${KO}$-dimension 
${s \equiv q-p \pmod{8}}$ (see, e.g., \cite{Marcolli-NoncomCosmology}) on ${V_{p,q}}$ such that
\begin{equation}
\left( { {\mathbb{C} \ell}_n \, , V_{p,q} \, , \Gamma , C} \right)
\end{equation}
satisfies all the axioms of a fermion space.\footnote{In the physics literature, the operator $C$ is called the \emph{charge conjugation} operator.} 

\begin{definition}
A \emph{matrix geometry of type ${(p,q)}$} is a finite real spectral triple $\left( { \mathcal{A}, \mathcal{H} , D , \gamma , J} \right)$, where the corresponding fermion space is given by:
\begin{itemize}
\item
${ \mathcal{A}= {{\mathrm{M}}_N (\mathbb{C})} }$
\item
${ \mathcal{H} = V_{p,q} \otimes {{\mathrm{M}}_N (\mathbb{C})} }$
\item
${ \langle {v \otimes A , u \otimes B} \rangle = \langle v , u \rangle \, \tr \left( A B^\ast \right)
\, , \quad v,u \in V_{p,q} \, , \  A,B \in {{\mathrm{M}}_N (\mathbb{C})} }$
\item
${ \pi (A) (v \otimes B) = v \otimes \left( AB \right)  }$
\item
${ \gamma (v \otimes A) = (\Gamma v) \otimes A }$
\item
${ J (v \otimes A) = (C v) \otimes A^\ast  \,}$.
\end{itemize}
\end{definition}
\noindent
The Dirac operators of type $(p,q)$ matrix geometries are expressed in term of gamma matrices ${\gamma^i \,}$, and commutators or anti-commutators with given Hermitian matrices $H$ and anti-Hermitian matrices $L$ (see \cite{Barrett-2015} and \cite{Barrett-Glaser-2016}). For a recent survey of interactions between fuzzy spectral triples and random matrix theory initiated in this paper we recommend \cite{hessam2022noncommutative}.


\section{Random matrix geometries of type \texorpdfstring{${(1,0)}$}{(1,0)}} 

In this section we describe a model for Euclidean quantum gravity on finite noncommutative spaces corresponding to the random matrix geometries of type ${(1,0) \,}$. The Dirac operator of type ${(1,0)}$ matrix geometries is given by \cite{Barrett-Glaser-2016}:
\begin{equation}
D = \{ H , \cdot \} \, , \quad H \in \mathcal{H}_N \, ,
\end{equation}
where ${\mathcal{H}_N}$ denotes the space of ${N \times N}$ Hermitian matrices. The Dirac operator ${D}$ acts on the Hilbert space ${\mathcal{H}= {\mathrm{M}}_N (\mathbb{C})}$ in the following way
\begin{equation}
D (B)= \{ H , B \} = HB + BH \, , \quad \forall B \in {\mathrm{M}}_N (\mathbb{C}) \,  .
\end{equation} 

The moduli space of Dirac operators is isomorphic to the space of Hermitian matrices 
${\mathcal{H}_N}$. A distribution of the form
\begin{equation}
\di \rho = e^{- \mathcal{S} (D)} \di D 
\end{equation}
is considered over ${\mathcal{H}_N}$, where
\begin{equation}
\di D \coloneqq \di H = \prod_{i=1}^N \di H_{ii} \, \prod_{1 \les i < j \les N} \di (\mathrm{Re} (H_{ij})) \, \di (\mathrm{Im} (H_{ij})) \, ,
\end{equation}
is the canonical Lebesgue measure on ${\mathcal{H}_N}$. Let us describe the action functional 
${\mathcal{S}(D) \,}$. Let
\begin{equation}
{{\mathbb{N}}^n_{\uparrow}} = \left\{ \vec{l} \in \mathbb{Z}^n \, | \, 1 \les l_1 \les l_2 \les \cdots \les l_n \right\}  .
\end{equation}
Suggested by Connes' spectral action, we define the action functional $\mathcal{S} (D)$ of the model by 
\begin{equation} \label{action-func-1}
\mathcal{S} (D) = {\mathcal{S}}_{\scaleto{\mathrm{unstable}}{4pt}} (D)
+ {\mathcal{S}}_{\scaleto{\mathrm{stable}}{4pt}} (D) \, ,
\end{equation}
where
\begin{equation} \label{unst-action-1}
{\mathcal{S}}_{\scaleto{\mathrm{unstable}}{4pt}} (D) = \tr \left( { {\mathcal{V}} (D) } \right) \, ,
\quad
{ {\mathcal{V}} (x) } = \frac{1}{2t} \Big( {\frac{x^2}{2} - \sum_{n =3}^d  {\alpha}_n \frac{x^n}{n} } \Big) \, ,
\end{equation}
and
\begin{equation}  \label{st-action-1}
{\mathcal{S}}_{\scaleto{\mathrm{stable}}{4pt}} (D)
= - \sum_{s=1}^{\mathfrak{g}} {\left( N/t \right)}^{-4s}
\sum_{n_I \in {{\mathbb{N}}^s_{\uparrow}}} {\hat{\alpha}}_{n_I} \,  
\prod_{i=1}^s \tr \left( {D^{n_i}} \right) \, .
\end{equation}
In the definition of the action functional, $t$ is a fixed parameter (``temperature''), the 
${{ \left( {\alpha_n , {\hat{\alpha}}_{n_I}} \right)}_{n , n_I}}$ are formal parameters, and, for each $s$, the summation over ${n_I \in {{\mathbb{N}}^s_{\uparrow}}}$ is a finite sum.\footnote{The formal parameters ${{ \left( {\alpha_n , {\hat{\alpha}}_{n_I}} \right)}_{n , n_I}}$ play the role of coupling constants in physics literature.} Let
\begin{equation}
\hbar \coloneqq \frac{t}{N} \, .
\end{equation}
For large $N$, the term ${{\mathcal{S}}_{\scaleto{\mathrm{stable}}{4pt}} (D)}$ can be considered as higher order terms in $\hbar$-expansion of the action functional ${\mathcal{S}(D) \, }$.

Using
\begin{equation}
D^n = \left( {H \otimes \mathds{1}_{{\left( {{\mathbb{C}}^{N}} \right)}^\ast} + 
\mathds{1}_{{\mathbb{C}}^{N}} \otimes H^{\mathrm{t}}} \right)^n
= \sum_{k=0}^n {\binom{n}{k}} H^{n-k} \otimes \left( {H^k} \right)^{\mathrm{t}} \, ,
\end{equation}
we get
\begin{align} \label{unst-action-2}
{\mathcal{S}}_{\scaleto{\mathrm{unstable}}{4pt}} (D) 
= \, &\frac{N}{t} \, \Big( {\frac{1}{2} \, {\tr \left( {H^2} \right)} - 
\sum_{n=3}^d \frac{\alpha_n}{n}  \, {\tr \left( {H^n} \right)}} \Big)
\nonumber \\
&- \frac{1}{2} \, \Big( {
- \frac{1}{t} \, {\left({ \tr (H)} \right)^2}
+ \sum_{n=3}^d  \frac{ \alpha_n}{nt} \, \sum_{r=1}^{n-1} 
{\binom{n}{r}} \tr{\left( H^{n-r} \right)} \, \tr{\left( {H^r} \right)}
} \Big) \, .
\end{align}
In addition, we have
\begin{align} \label{prod-tr-D}
&\prod_{i=1}^s \tr \left( {D^{n_i}} \right) =
\nonumber \\
&\sum_{m=0}^s {\left(2N\right)}^{s-m}
 {\sum_{\substack {J \subseteq I \\ |J| = m}}  \prod_{i \in I \backslash J} \tr{\left( H^{n_i} \right)}} 
\, \Big( { \sum_{\substack {(r_{j_1} , \cdots , r_{j_m}) \\ 1 \les r_j \les n_j -1 \, , \ j \in J}}
\prod_{j \in J} 
{n_j \choose r_j} \tr{\left( H^{n_j - r_j} \right)} \, \tr{\left( {H^{r_j}} \right)} } \Big) \, ,
\end{align}
where ${I = \{ 1 , \cdots , s \} \,}$.
By substituting \eqref{unst-action-2} and \eqref{prod-tr-D} into \eqref{action-func-1}, we get the expression for the action functional ${\mathcal{S} (D)}$ in terms of the spectrum of the Hermitian matrix 
${H \in \mathcal{H}_N \,}$.


\section{Topological expansion of the action functional} 

In the following, we rewrite the action functional ${\mathcal{S} (D) \,}$, in a succinct form, as a summation over a finite set of (properly defined) equivalence classes of surfaces.  

We start by recalling the notion of a surface with polygonal boundary. A compact, connected, oriented surface $C$ is referred to have $n$ polygonal boundary components of perimeters 
${ {\{ \ell_i \}}_{i=1}^n \, , \, \ell_i \gre 1 \, ,}$ if, for each ${1 \les i \les n \,}$, the $i$-th connected component of ${\partial C}$ is equipped with a cellular decomposition into ${\ell_i}$ 0-cells and ${\ell_i}$ 1-cells. We refer to the 1-cells in each connected component of ${\partial C}$ as the \emph{sides} of that polygon. We define an equivalence relation between surfaces with polygonal boundaries in the following way: 
\begin{definition} \label{elementary-2-cell}
Two compact, connected, oriented surfaces ${C_1 \, , \, C_2}$ with polygonal boundaries are considered equivalent if there exists an orientation-preserving diffeomorphism ${F : C_1 \to C_2}$ which restricts to a \emph{cellular} homeomorphism ${f : \partial C_1 \to \partial C_2}$ whose inverse is also a cellular map.
\end{definition}
The set $\widehat{\mathcal{C}}$ of equivalence classes of compact, connected, oriented surfaces with polygonal boundaries is in bijective correspondence with the set
\begin{equation}
\left\{  { (  g  ; \vec{\ell} \, ) \, \big|  \, g \gre 0 \, , \,
\vec{\ell} \in  {{\mathbb{N}}^n_{\uparrow}} \, , \, n \gre 1}  \right\}  ,
\end{equation}
where $g$ denotes the genus of the corresponding closed surface. Inspired by \cite{Borot-2014}, we refer to the combinatorial data ${ (  g; \vec{\ell} \, ) }$, and its corresponding equivalence class 
${[C] \in {\widehat{\mathcal{C}}}}$ of surfaces with polygonal boundaries, as the \emph{elementary 2-cell} of type ${ (  g; \vec{\ell} \, ) }$. 

We isolate the \emph{free} part of the action functional and denote it by
\begin{equation}
\mathcal{S}_0 (H) = \frac{N}{2 t} \, {\tr \left( {H^2} \right)} \, .
\end{equation}
Let
\begin{equation} \label{P_l}
P_\ell (H) \coloneqq \frac{\tr \left( H^\ell \right)}{\ell} \, , \quad \forall \ell \in \mathbb{N} 
\, , \ H \in \mathcal{H}_N \, .
\end{equation}
Consider the following two sets:
\begin{equation}
{\mathfrak{L}}_{\scaleto{\mathrm{disk}}{4pt}} =
\left\{  \ell \in \mathbb{N} \, | \,  3 \les \ell \les d \right\}  ,
\end{equation}
\begin{equation}
{\mathfrak{L}}_{\scaleto{\mathrm{cylinder}}{4pt}} =
\left\{ { ( \ell_1 , \ell_2 ) \in \mathbb{N}^2_{\uparrow} \, | \,  2 \les \ell_1 + \ell_2 \les d  } \right\} .
\end{equation}
We rewrite ${{\mathcal{S}}_{\scaleto{\mathrm{unstable}}{4pt}} (D) }$ in the following from
\begin{align} \label{unst-action-3}
{\mathcal{S}}_{\scaleto{\mathrm{unstable}}{4pt}} (D) 
= \mathcal{S}_0 (H) &- \frac{N}{t} 
\sum_{  \ell  \in {\mathfrak{L}}_{\scaleto{\mathrm{disk}}{4pt}}}
\rt^{(0)}_\ell \, P_\ell (H)
\nonumber \\
&- \frac{1}{2} \, \sum_{ (\ell_1 , \ell_2) \in {\mathfrak{L}}_{\scaleto{\mathrm{cylinder}}{4pt}}}
\rt^{(0)}_{\ell_1 , \ell_2} \, P_{\ell_1} (H) \, P_{\ell_2} (H) \, ,
\end{align}
where ${\rt^{(0)}_\ell = \alpha_\ell \, , \, 3 \les \ell \les d \,}$; ${\rt^{(0)}_{1,1} = -1/t \,}$, and for 
${ ( \ell_1 , \ell_2 ) \in {\mathfrak{L}}_{\scaleto{\mathrm{cylinder}}{4pt}} \backslash 
\left\{ ( 1,1)  \right\} \,}$, ${\rt^{(0)}_{\ell_1 , \ell_2}}$ is an integral multiple of 
${\alpha_{\left( \ell_1 + \ell_2 \right)} /t \,}$.

For each ${1 \les s \les \mathfrak{g} \,}$, and ${0 \les m \les s \,}$, let
\begin{equation}
\Upsilon_{s,m} =
\bigcup_{\substack{ n_I \, , \, |I| =s \\ J \subseteq I \, , \, |J| = m \\
1 \les r_j \les n_j -1 \, ,\, j \in J }}
\left\{ \left( {r_J \, , \, n_J - r_J \, , \, n_{I \backslash J}} \right) \right\} \big/ \sim \, ,
\end{equation}
where two ${(s+m)}$-tuples are considered equivalent if there exists a permutation 
$\sigma \in {\mathfrak{S}}_{s+m}$ which maps one to the other. Consider the set
\begin{equation}
{\mathfrak{L}}_{s,m} = 
\left\{ \vec{\ell} \in {\mathbb{N}}^{s+m}_{\uparrow} \, \big| \, 
[ \sigma \cdot \vec{\ell} \, ]  \in \Upsilon_{s,m} \  \ \text{for some} \ 
\sigma \in {\mathfrak{S}}_{s+m} \right\}  .
\end{equation}
We rewrite ${{\mathcal{S}}_{\scaleto{\mathrm{stable}}{4pt}} (D) }$ in the following form
\begin{equation} \label{st-action-3}
{\mathcal{S}}_{\scaleto{\mathrm{stable}}{4pt}} (D) =
- \sum_{\substack{ 1 \les s \les \mathfrak{g} \\ 0 \les m \les s}} 
\frac{ {\left( N/t \right)}^{2- 2 (s+1) - (s+m)}}{(s+m)!}
\sum_{ \vec{\ell} \in {\mathfrak{L}}_{s,m}}
{\rt}^{(s+1)}_{\vec{\ell}} \, \prod_{i=1}^{s+m} P_{\ell_i} (H) \, ,
\end{equation}
where, for each ${ \vec{\ell} \in {\mathfrak{L}}_{s,m}} \,$, 
${{\rt}^{(s+1)}_{\vec{\ell}} = t^{s-m} \, \tilde{\alpha}_{\vec{\ell}} \,}$, and ${\tilde{\alpha}_{\vec{\ell}}}$ is a finite linear combination of the formal parameters ${{\hat{\alpha}}_{n_I}}$'s with integral coefficients.

Consider the following two sets of elementary 2-cells:
\begin{equation}
{\mathcal{C}}_{\scaleto{\mathrm{unstable}}{4pt}}
= \left\{ {(0 ; \vec{\ell} \,) \, \big| \, \vec{\ell} \in {{\mathfrak{L}}_{\scaleto{\mathrm{disk}}{4pt}} 
\cup {{\mathfrak{L}}_{\scaleto{\mathrm{cylinder}}{4pt}}}} } \right\}  ,
\end{equation}
\begin{equation}
{\mathcal{C}}_{\scaleto{\mathrm{stable}}{4pt}}
= \bigcup_{\substack{1 \les s \les \mathfrak{g} \\ 0 \les m \les s}}
 \left\{ {(s+1 ; \vec{\ell} \,) \, \big| \, \vec{\ell} \in {{\mathfrak{L}}_{s,m}} } \right\} .
\end{equation}
We identify the set
\begin{equation} \label{ele-2-cells}
\mathcal{C} = {\mathcal{C}}_{\scaleto{\mathrm{unstable}}{4pt}} \cup 
{\mathcal{C}}_{\scaleto{\mathrm{stable}}{4pt}} 
\end{equation}
with the corresponding set of equivalence classes ${[C]}$ of surfaces with polygonal boundaries. We assign a \emph{Boltzmann weight}, equal to ${{\rt}^{(g)}_{\vec{\ell}} }$, to each elementary 2-cell 
${[C] \in \mathcal{C}}$ of type ${(g; \vec{\ell} \,) \,}$. Note that the elementary 2-cells in 
${{\mathcal{C}}_{\scaleto{\mathrm{unstable}}{4pt}}}$ (resp. 
${{\mathcal{C}}_{\scaleto{\mathrm{stable}}{4pt}}}$) are represented by surfaces whose Euler characteristic satisfies ${\chi \gre 0}$ (resp. ${\chi < 0}$). For each elementary 2-cell 
${[C] \in \mathcal{C}}$ of type ${(g; \vec{\ell} \,) }$ with Boltzmann weight 
${{\rt}^{(g)}_{\vec{\ell}} \, , \ \vec{\ell} \in \mathbb{N}^n_{\uparrow} \, }$, let
\begin{equation} \label{T_C-H}
T_{\scaleto{[C]}{6pt}} (H) \coloneqq {\rt}^{(g)}_{\vec{\ell}} \, \prod_{i=1}^n P_{\ell_i} (H) \, , \quad 
H \in \mathcal{H}_N \, ,
\end{equation}
where ${P_{\ell} (H)}$ is defined by \eqref{P_l}. We rewrite 
\eqref{unst-action-3} and \eqref{st-action-3}, respectively, in the following from
\begin{equation}
{\mathcal{S}}_{\scaleto{\mathrm{unstable}}{4pt}} (D)
= \mathcal{S}_0 (H) - \sum_{[C] \in {{\mathcal{C}}_{\scaleto{\mathrm{unstable}}{4pt}}}}
\frac{{\left( N/t \right)}^{\chi (C)}}{\left( { \beta_0 \left( { \partial C} \right)} \right) !} \ 
T_{\scaleto{[C]}{6pt}} (H) \, ,
\end{equation}
\begin{equation}
{\mathcal{S}}_{\scaleto{\mathrm{stable}}{4pt}} (D)
=  - \sum_{[C] \in {{\mathcal{C}}_{\scaleto{\mathrm{stable}}{4pt}}}}
\frac{{\left( N/t \right)}^{\chi (C)}}{\left( { \beta_0 \left( { \partial C} \right)} \right) !} \ 
T_{\scaleto{[C]}{6pt}} (H) \, ,
\end{equation}
where ${ \beta_0 \left( { \partial C} \right)}$ denotes the zeroth Betti number, i.e. the number of connected components, of the boundary ${\partial C}$ of a surface $C \,$. Thus, we get:

\begin{proposition}
The action functional ${\mathcal{S}(D)}$ for the random matrix geometries of type ${(1,0) \,}$, given by \eqref{action-func-1}, can be decomposed in the following form
\begin{equation} \label{action-func-2}
\mathcal{S}(D) = \mathcal{S}_0 (H) + \mathcal{S}_{\mathrm{int}} (H) \, ,
\end{equation}
where 
\begin{equation}
 \mathcal{S}_{\mathrm{int}} (H)
=  - \sum_{[C] \in {\mathcal{C}}}
\frac{{\left( N/t \right)}^{\chi (C)}}{\left( { \beta_0 \left( { \partial C} \right)} \right) !} \ 
T_{\scaleto{[C]}{6pt}} (H) \, ,
\end{equation}
and $\mathcal{C}$ is given by \eqref{ele-2-cells}.
\end{proposition}


\section{The corresponding 1-Hermitian matrix model}
From now on, we consider the multi-trace 1-Hermitian matrix model corresponding to the random matrix geometries of type ${(1,0)}$ with the distribution ${\di \rho = e^{- \mathcal{S} (D)} \di D \,}$, in the sense of \emph{formal} matrix integrals. In other words, we treat the term ${\mathcal{S}_{\mathrm{int}} (H)}$ in \eqref{action-func-2} as a perturbation of ${\mathcal{S}_0 (H) \,}$. 

Consider the normalized Gaussian measure 
\begin{equation} \label{Gaussian}
{\di {\rho}_0 } = c \, e^{- \mathcal{S}_0 (H)} \di H 
= c \, \exp \Big( -\frac{N}{2 t} \, \tr  \left( H^2 \right) \Big) \di H 
\end{equation}
over ${\mathcal{H}_N}$ with total mass one. Here
\begin{equation}
c ={\left[ 2^{N} \, {\left( {\frac{\pi t}{N}} \right)}^{{N^2}} \, \right]}^{- 1/2}\, .
\end{equation}
Denote by $\ft$ the sequence of Boltzmann weights ${\rt^{(g)}_{\vec{\ell}}}$ in 
${\mathcal{S}_{\mathrm{int}} (H) \,}$. We consider
\begin{equation}
\Phi (H) = \exp \left( {- \mathcal{S}_{\mathrm{int}} (H)} \right)
\end{equation}
as a formal power series in $\ft \, $, i.e. an exponential generating function. The \emph{partition function} $Z_N$ of the model is defined by
\begin{equation}
Z_N = \rho_0 \left[ \Phi (H) \right] \
{\mathrel{\overset{\makebox[0pt]{\mbox{\normalfont\tiny\sffamily formal}}}{=}}} \ \
c  \int_{\mathcal{H}_N} \Phi (H) \, e^{- \mathcal{S}_0 (H)} \di H  \, ,
\end{equation}
where the second integral is understood in the sense that we expand ${\Phi (H)}$ as a power series in $\ft \,$, and \emph{interchange} the integration with the summation. 
 
The \emph{disconnected $n$-point correlators} ${\widehat{W}_n (x_1 , \cdots , x_n) \,}$, ${n\gre 1 \,}$, of the model are defined as the joint moments of 
\begin{equation}
\mathrm{X}_j = \tr \left( {({x_j}{\mathds{1}_N} -H)^{-1}} \right) \, ,
\quad {x_j \in \mathbb{C}\backslash \text{Spec}(H) \,} ,
\end{equation}
i.e.
\begin{align} \label{discon-cor-1}
\widehat{W}_n (x_1 , \cdots , x_n) &=
\frac{1}{Z_N} \, \rho_0 
\Big[ { \Phi (H) \prod_{j=1}^n \tr \left( {({x_j}{\mathds{1}_N} -H)^{-1}} \right)} \Big]
\nonumber \\
&\ {\mathrel{\overset{\makebox[0pt]{\mbox{\normalfont\tiny\sffamily formal}}}{=}}} \ \ 
\mathbb{E} \Big[  \,{\prod_{j=1}^n \tr \left( {({x_j}{\mathds{1}_N} -H)^{-1}} \right)} \Big] \, ,
\end{align}
where ${{({x_j}{\mathds{1}_N} -H)^{-1}}}$ denotes the resolvent of $H \,$.\footnote{Strictly speaking, in the context of formal matrix integrals, one works with the formal series
\begin{equation} \label{discon-cor-2}
{\widetilde{\mathrm{X}}}_j = \sum_{\ell=0}^{\infty} \frac{\tr (H^\ell)}{x_j^{\ell+1}}
\end{equation}
instead of ${\tr \left( {({x_j}{\mathds{1}_N} -H)^{-1}} \right) \,}$.} The \emph{connected $n$-point correlators} ${{W}_n (x_1 , \cdots , x_n) \,}$, ${n\gre 1 \,}$, of the model are defined as the joint cumulants of ${\mathrm{X}_j}$'s, i.e.
\begin{equation} \label{connected-cor-1}
W_n (x_1 , \cdots , x_n) = \sum_{K \vdash \llbracket 1,n \rrbracket} (-1)^{[K] -1} \, ({[K] -1}) ! \, \prod_{i=1}^{[K]} 
\widehat{W}_{| {K_i} |} \left( {x_{K_i}} \right) \, .
\end{equation} 
In \eqref{connected-cor-1}, the sum runs over partitions of 
${\llbracket 1, n \rrbracket \coloneqq \{ 1,2,3, \cdots , n \} \,}$, the number of subsets in a partition $K$ is denoted by $[K] \,$, and 
\begin{equation}
\widehat{W}_{| {K_i} |} \left( {x_{K_i}} \right) \equiv
\widehat{W}_{| {K_i} |} \left( {(x_j)}_{j \in K_i} \right) \, .
\end{equation}


\subsection{Topological expansion of the correlators}
\begin{proposition}
The connected $n$-point correlators ${{W}_n (x_1 , \cdots , x_n)}$ of the random matrix geometries of type ${(1,0)}$ with the distribution ${\di \rho = e^{- \mathcal{S} (D)} \di D }$ have a 
large $N$ expansion of \emph{topological type}, given by
\begin{equation} \label{topological-expansion-1}
W_n (x_1 , \cdots , x_n) = \sum_{g \gre 0} { \left( {N/t} \right)}^{2 - 2g -n} \
W_{g,n} (x_1 , \cdots , x_n) \, ,
\end{equation}
where ${W_{g,n} (x_1 , \cdots , x_n)}$ is defined, in the following, by \eqref{W_g,n-1}.\footnote{One should \emph{not} misinterpret \eqref{topological-expansion-1} as the \emph{asymptotic expansion} of the connected correlators as ${N \to \infty}$ 
(see \cite{Eynard-CountingSurfaces}, \cite{Borot-Guionnet-Kozlowski-2015}).}
\end{proposition}
\begin{proof}
We use the \emph{Wick's Theorem} and the techniques of \cite{Brezin-Itzykson-Parisi-Zuber-1987} to relate the formal matrix integrals in our model to the combinatorics of \emph{stuffed maps} \cite{Borot-2014}, \cite{Borot-Shadrin-2017}. Considering \eqref{discon-cor-1} and \eqref{discon-cor-2}, the computation of ${{\widehat{W}}_n (x_1 , \cdots , x_n)}$ leads to Gaussian integrals of the following form
\begin{equation} \label{stuffed-map-1}
\rho_0 \Bigg[ \, { \prod_{j=1}^n \frac{\tr (H^{\ell_j})}{x_j^{\ell_j +1}} \,
\prod_{i=1}^{| \mathcal{C} |} \frac{1}{k_i !} \,
\Big( { \frac{{\left( N/t \right)}^{\chi (C_i)}}{\left( { \beta_0 \left( { \partial C_i} \right)} \right) !} \ 
T_{\scaleto{[C_i]}{6pt}} (H) } \Big)^{k_i} } \Bigg] \, .
\end{equation}
To compute \eqref{stuffed-map-1} using Wick's Theorem, we represent each term of the form\\ ${{\tr (H^{\ell_j})}/{x_j^{\ell_j +1}}}$, in \eqref{stuffed-map-1}, by a marked face of perimeter ${\ell_j}$ with Boltzmann weight ${x_j^{- (\ell_j + 1)} \,}$, ${\, j=1 , \cdots , n \,}$.\footnote{A polygon of perimeter 
$\ell$ is an oriented 2-dimensional CW complex consisting of a 2-cell, homeomorphic to a disk, whose boundary is equipped with a cellular decomposition into $\ell$ 0-cells and $\ell$ 1-cells. A polygon is called rooted if we distinguish one of the 0-cells and its incident 1-cell on the boundary. We refer to a labeled rooted polygon of perimeter $\ell$ as a marked face of perimeter $\ell \,$.}  Also, we represent each term of the form
\begin{equation}
{ \frac{{\left( N/t \right)}^{\chi (C_i)}}{\left( { \beta_0 \left( { \partial C_i} \right)} \right) !} \ 
T_{\scaleto{[C_i]}{6pt}} (H) } \, , 
\end{equation}
in \eqref{stuffed-map-1}, by a surface ${C_i}$ of Euler characteristic ${{\chi (C_i)}}$ and Boltzmann weight ${\rt^{(g_i)}_{\vec{\ell_i}}}$ representing the corresponding elementary 2-cell ${[C_i] \in \mathcal{C}}$ of type ${(g_i ; \vec{\ell_i}) \,}$, ${i=1 , \cdots ,  | \mathcal{C} | \, }$. The orientation on each ${C_i}$ induces an orientation on ${\partial C_i \,}$. 

Let $\Xi$ be the collection of all surfaces representing the terms in \eqref{stuffed-map-1} in the above-mentioned way. Consider a pairing $\sigma$ on the set of sides of the connected components of the boundary of surfaces in $\Xi \,$. We glue the surfaces in $\Xi$ along the sides of their boundary, according to $\sigma$, such that the gluing map reverses the orientation. The resulting stuffed map 
${M=(S,G)}$ consists of an oriented, not necessarily connected surface $S \,$, and a graph $G$ embedded into $S \,$.\footnote{If each connected component of ${S \backslash G}$ is homeomorphic to an open disk, then ${M=(S,G)}$ is called a \emph{map}.} We denote by $\tilde{S}$ the surface which one gets by deleting the marked faces from $S \,$.

It can be shown that each vertex (resp. edge) of $G$ contributes a weight $N$ 
(resp. ${t/N}$) \cite{Brezin-Itzykson-Parisi-Zuber-1987}. In addition, each unmarked connected component $\mathfrak{U}$ of ${S \backslash G}$ contributes a weight 
${(N /t)^{ \chi (\mathfrak{U})} \,}$. Hence, the exponent of $N \,$, in the total contribution corresponding to ${M= (S,G) \,}$, equals ${\chi (\tilde{S}) \,}$.\footnote{This fact about the exponent of $N \,$, in the case of maps (or, equivalently, \emph{ribbon graphs}), was first noticed by Gerard 't Hooft in \cite{tHooft-1974} (see, e.g., \cite{Eynard-CountingSurfaces}).} 

In addition to the pre-mentioned Boltzmann weights assigned to the 2-cells of $M$, we assign a Boltzmann weight equal to $t$ to each vertex of $G \,$. The total Boltzmann weight of the isomorphism class ${[M]}$ of a stuffed map $M \,$, denoted by ${\mathfrak{Bw} ([M])}$, is defined to be the product of all Boltzmann weights assigned to the cells in $M$ divided by the order ${{| {\text{Aut} (M)} |}}$ of the automorphism group of $M$. The contribution of the isomorphism class ${[M]}$ of a Boltzmann-weighted not necessarily connected stuffed map ${M=(S,G)}$ to ${{\widehat{W}}_n (x_1 , \cdots , x_n)}$ is given by
\begin{equation} \label{discon-stuf-map-1}
{\left( {N/t} \right)}^{\chi (\tilde{S})} \ {\mathfrak{Bw} ([M])} \, .
\end{equation}

Let ${{\mathbb{M}}_{g,n} \left( {\mathcal{C}} \right)}$ be the set of isomorphism classes of the Boltzmann-weighted connected closed stuffed maps ${M=(S,G)}$ of genus $g$ with $n$ marked faces, such that the equivalence class of each unmarked connected component of ${S \backslash G}$ (in the sense of Definition \ref{elementary-2-cell}) is in $\mathcal{C} \,$. Considering \eqref{discon-stuf-map-1}, and
\begin{equation} \label{discon-con-correlator}
{{\widehat{W}}_n (x_1 , \cdots , x_n)} = \sum_{K \vdash \llbracket 1,n \rrbracket}
\prod_{i=1}^{[K]} 
{W}_{| {K_i} |} \left( {x_{K_i}} \right) \, ,
\end{equation}
we have
\begin{equation}
{{{W}}_n (x_1 , \cdots , x_n)} = \sum_{g \gre 0} \ 
\sum_{[M] \in {{\mathbb{M}}_{g,n} \left( {\mathcal{C}} \right)}}
{\left( {N/t} \right)}^{\chi (\tilde{S})} \ {\mathfrak{Bw} ([M])} \, .
\end{equation}
Since, for a connected closed stuffed map ${M=(S,G)}$ of genus $g$ with $n$ marked faces,
\begin{equation}
{\chi (\tilde{S})} = 2 - 2g -n \, ,
 \end{equation} 
we get \eqref{topological-expansion-1}, where 
\begin{equation} \label{W_g,n-1}
{{{W}}_{g,n} (x_1 , \cdots , x_n)} = 
\sum_{[M] \in {{\mathbb{M}}_{g,n} \left( {\mathcal{C}} \right)}}
{\mathfrak{Bw} ([M])} \ \in \mathbb{Q}[\ft][[t]][[{(x_j^{-1})}_j]] \, .
\end{equation}
\end{proof}

Let ${{{\mathbb{M}}_{g,n; \vec{\ell}} \left( {\mathcal{C}} \right)} \, ,}$ 
${\vec{\ell} = (\ell_1 , \cdots , \ell_n) \in \mathbb{N}^n \,}$, be the set of isomorphism classes 
$[M] \in {\mathbb{M}}_{g,n} \left( {\mathcal{C}} \right)$ of the stuffed maps with $n$ marked faces of perimeters ${\ell_j \, ,}$ ${j=1 , \cdots , n \, }$. By \eqref{W_g,n-1}, the generating series 
${Q_{g ; \vec{\ell}} \in \mathbb{Q}[\ft][[t]]}$ of the stuffed maps, corresponding to our model, of genus $g$ with $n$ polygonal boundaries of perimeters ${\ell_j \, ,}$ ${j=1 , \cdots , n \, }$, satisfies 
\begin{align} \label{disk-stuffed-map-2}
{{{W}}_{g,n} (x_1 , \cdots , x_n)} &= \delta_{g,0} \, \delta_{n,1} \, \frac{t}{x_1} +
\sum_{\vec{\ell} \in \mathbb{N}^n} \
\sum_{[M] \in {{{\mathbb{M}}_{g,n; \vec{\ell}} \left( {\mathcal{C}} \right)}}}
{\mathfrak{Bw} ([M])}
\nonumber \\
&= \delta_{g,0} \, \delta_{n,1} \, \frac{t}{x_1} + \sum_{\vec{\ell} \in \mathbb{N}^n}
{Q_{g ; \vec{\ell}}} \ { \prod_{j=1}^n x_j^{-(\ell_j +1)}}   \, .
\end{align}


\subsection{Large-\texorpdfstring{$N$}{N} spectral distribution}
Using \eqref{disk-stuffed-map-2}, we see that the generating series ${Q_{0; \ell} \, ,}$ ${\ell \in \mathbb{N} \,}$, of the rooted planar stuffed maps with topology of a disk and perimeter $\ell $, is given by
\begin{equation} \label{disk-stuffed-map}
Q_{0; \ell} = 
\sum_{[M] \in {{{\mathbb{M}}_{0,1; \ell} 
\left( {{\mathcal{C}}_{\scaleto{\mathrm{unstable}}{4pt}}} \right)}}}
{\mathfrak{Bw} ([M])} \, x^\ell \ \in \mathbb{Q}[{\mathfrak{t}^{(0)}_{\vec{\ell}}}][[t]] \, ,
\end{equation}
where ${\mathfrak{t}^{(0)}_{\vec{\ell}}}$ denotes the sequence of Boltzmann weights 
${\rt^{(0)}_{\vec{\ell}} \, ,}$ ${\vec{\ell} \in {{\mathfrak{L}}_{\scaleto{\mathrm{disk}}{4pt}} \cup {{\mathfrak{L}}_{\scaleto{\mathrm{cylinder}}{4pt}}}} \,}$. 

If the Boltzmann weights ${\mathfrak{t}^{(0)}_{\vec{\ell}}}$ (or, equivalently, the formal parameters ${\alpha_n}$, ${3 \les n \les d}$) have given values, then there exists a critical temperature ${t_c >0}$ such that, for any ${|t| < t_c \,}$, we have ${Q_{0; \ell} < \infty \, , \ \forall \ell \in \mathbb{N}}$ \cite{Borot-2014}. From now on, we restrict ourselves to the case ${0< t < t_c \,}$, where ${t_c}$ is specified according to each set of given values to ${\alpha_n}$'s. Hence, we have the following \emph{one-cut Lemma} \cite{Borot-2014}, \cite{Borot-Bouttier-Guitter-2012}:
\begin{lemma} \label{one-cut lemma}
For given values to the Boltzmann weights ${\mathfrak{t}^{(0)}_{\vec{\ell}} \,}$, and ${0< t < t_c \,}$, the series
\begin{equation} \label{Laurent-exp-W01}
W_{0,1} (x) = \frac{t}{x} + \sum_{\ell=1}^ \infty \frac{Q_{0; \ell}}{x^{\ell+1}} \, ,
\end{equation}
is the Laurent expansion at ${x = \infty}$ of a holomorphic function, denoted by ${W_{0,1} (x) \,}$, on 
${\mathbb{C} \backslash \Gamma \,}$, where ${\Gamma = [ \mathfrak{a} ,\mathfrak{b}] \subset \mathbb{R}}$ depends on ${\mathfrak{t}^{(0)}_{\vec{\ell}} \, , t \,}$. The limits 
${\lim_{\epsilon \to 0^+} W_{0,1} (s \pm \ci \epsilon) \, ,}$ ${\forall s \in \Gamma^{\mathrm{o}}}$ exist, and the jump discontinuity
\begin{equation} \label{jump-discon-1}
\varphi(s)= \frac{1}{2 \pi \ci}   \lim_{\epsilon \to 0^+} \left( {W_{0,1} (s - \ci \epsilon)}
- {W_{0,1} (s + \ci \epsilon)} \right) 
\end{equation}
assumes positive values on the interior ${ \Gamma^{\mathrm{o}}}$ of the discontinuity locus $\Gamma$, and vanishes at ${\partial \Gamma }$.
\end{lemma}

Consider the measure ${\mu = \varphi (s) \di s}$ on $\mathbb{R \,}$, where ${\di s}$ denotes the Lebesgue measure, and ${\varphi (s)}$ is given by \eqref{jump-discon-1}. By the \emph{Sokhotski-Plemelj Theorem}, the function ${W_{0,1} (x)}$ is, indeed, the \emph{Stieltjes transform} 
${\mathscr{S} [\mu ] (x)}$ of $\mu \,$, i.e.
\begin{equation} \label{Stielt-eq-measure-1}
{W_{0,1} (x)} = {\mathscr{S} [\mu ] (x)} = \int \frac{\varphi (s)}{x-s} \di s \, , \quad x \in \mathbb{C}\backslash \text{supp}(\mu) \, .
\end{equation}
In addition, consider the \emph{empirical spectral distribution (empirical measure)}
\begin{equation}
{\mu_N =\frac{1}{N} \sum_{i=1}^N \delta_{\lambda_i}}
\end{equation}
on $\mathbb{R} \,$, where ${{\{ \lambda_i \}}_{i=1}^{N} \, }$, ${\lambda_i \in \mathbb{R} \,}$, denotes the eigenvalues of the random Hermitian matrix ${H \in \mathcal{H}_N}$ corresponding to the random matrix geometries of type ${(1,0)}$ with the distribution ${\di \rho = e^{- \mathcal{S} (D)} \di D \, }$.\footnote{The Dirac measure at ${\lambda \in \mathbb{R}}$ is denoted by ${\delta_{\lambda} \,}$.} By \eqref{discon-cor-1}, we have
\begin{equation} \label{Stielt-eq-measure-2}
\frac{1}{N} \, W_1 (x) = \mathbb{E} \left[  \mathscr{S} \left[ \mu_N \right] (x) \right]  \, .
\end{equation}

Motivated by \eqref{topological-expansion-1}, we assume that
\begin{equation}
\lim_{N \to \infty} \frac{1}{N} \, W_1 (x) = {W_{0,1} (x)}  \, .
\end{equation}
Therefore, considering \eqref{Stielt-eq-measure-1} and \eqref{Stielt-eq-measure-2}, the expected distribution of the eigenvalues ${{\{ \lambda_i \}}_{i=1}^{N} }$ is given by 
\begin{equation}
{\mu = \varphi (s) \di s \,} ,
\end{equation}
up to terms with exponential decay, as ${N \to \infty \,}$.\footnote{The measure ${\mu = \varphi (s) \di s}$ plays the same role as what is called, in the context of \emph{convergent} matrix integrals, the \emph{equilibrium measure} (see, e.g., \cite{Anderson-Guionnet-Zeitouni-RMT}).} We refer to the measure ${\mu = \varphi (s) \di s}$ as the \emph{large-$N$ spectral distribution}. In addition, from now on, we assume that the sequence of Boltzmann weights $\ft \,$, and the parameter $t \,$, are \emph{tame}, in the sense of \cite{Borot-2014}, Definition 4.1. Hence, each ${{{W}}_{g,n} (x_1 , \cdots , x_n) \,}$, a priori defined as the generating series of stuffed maps, upgrades to a holomorphic function on 
${{ \left( {\mathbb{C} \backslash \Gamma} \right)}^n}$ which has a jump discontinuity when one of $x_i$'s crosses $\Gamma \,$.


\section{Schwinger-Dyson equations}

Our main tool for analyzing the ${{{W}}_{g,n} (x_1 , \cdots , x_n)}$'s is an infinite system of equations, called the \emph{Schwinger-Dyson equations} (SDEs), satisfied by the $n$-point correlators of the model. In the matrix model framework, they were introduced by Migdal \cite{Migdal-1983}, and referred to as the \emph{loop equations}. There are several versions of SDEs for matrix models (and some other closely related models in statistical physics, e.g., the \emph{$\beta$-ensembles}) in the literature 
(see, e.g., \cite{Eynard-CountingSurfaces}, \cite{Borot-Guionnet-2013}). However, the root of all of them is the invariance of the integral of a top degree differential form under a 1-parameter family of orientation-preserving diffeomorphisms on a manifold. 

To put the above-mentioned differential geometric fact in a precise form, consider an oriented connected Riemannian $n$-manifold $M$ with the Riemannian volume form $\omega \,$. Let $V$ be a smooth vector field on $M$ with a local flow ${\phi_t : M \to M \,}$. Consider a smooth function ${\Psi : M \to \mathbb{R} \,}$. Let ${\Omega \subset M}$ be a compact $n$-dimensional submanifold of $M$. Since
\begin{equation} 
\int_{\phi_t (\Omega)}  {\Psi \omega}  = \int_{\Omega} {\phi_t}^\ast  \left( {\Psi \omega} \right)
\, , \quad \forall t \in (-\epsilon , \epsilon) \, ,
\end{equation}
using Cartan's magic formula and Stokes' theorem, we get
\begin{equation} \label{Integral invariance-1}
 \int_{\Omega}  \left( {d \Psi (V) + \Psi \, div(V)} \right) \omega
 - \int_{\partial \Omega} \Psi  (\iota_V \omega) = 0 \, .
\end{equation}
In \eqref{Integral invariance-1}, the \emph{exterior derivative}, the \emph{interior product} by $V$, and
the \emph{divergence} of $V$ are denoted by $d \,$, ${\iota_V \,}$, and ${div(V) \,}$, respectively.  

The Schwinger-Dyson equations for the multi-trace 1-Hermitian matrix models are derived in \cite{Borot-2014}, using \emph{Tutte's decomposition} applied to the stuffed maps. In this section, we give a proof of them based on the above-mentioned differential geometric fact.\footnote{In general, there are, at least, two approaches in the literature for deriving the SDEs for matrix models: one is usually referred to as ``integration by parts'' or ``invariance under change of variable'' which is basically the above-mentioned differential geometric approach; and the other one is combinatorial and based on Tutte's 
decomposition. The differential geometric approach has the advantage that it can be applied to a slightly more general class of models, e.g. the $\beta$-ensembles for arbitrary ${\beta >0  \,}$, which do not necessarily have a combinatorial interpretation.}   

Consider the action of the unitary group ${{\mathrm{U}}_N}$ on ${\mathcal{H}_N}$ by conjugation, i.e. $u \cdot H \coloneqq u H u^{-1}$ for ${u \in {\mathrm{U}}_N}$ and ${H \in {\mathcal{H}_N}}$. Since the action functional ${\mathcal{S} (D) : \mathcal{H}_N \to \mathbb{R}}$ is invariant under the pre-mentioned action of ${{\mathrm{U}}_N}$ on ${\mathcal{H}_N \,}$, we can rewrite ${\mathcal{S} (D)}$ as a function of the eigenvalues ${{\{ \lambda_i \}}_{i=1}^{N} \, }$, ${\lambda_i \in \mathbb{R} \,}$, of the random Hermitian matrix ${H \in \mathcal{H}_N \,}$. 

Let
\begin{equation}
(\emm)^k \coloneqq \underbrace{\emm \times \cdots \times \emm}_{k\text{-times}}
\end{equation}
be the product measure on ${\mathbb{R}^k}$ corresponding to the unnormalized empirical measure
\begin{equation}
\emm = \sum_{i=1}^N \delta_{\lambda_i} 
\end{equation}
on $\mathbb{R} \,$. For each elementary 2-cell 
${[C] \in \mathcal{C}}$ of type ${(g; \vec{\ell} \,) \, ,}$ ${\vec{\ell} \in \mathbb{N}^k_{\uparrow} \,}$, we rewrite ${T_{\scaleto{[C]}{6pt}} (H) \,}$, defined by \eqref{T_C-H}, in the following from
\begin{align}
T_{\scaleto{[C]}{6pt}} (H) &= {\rt}^{(g)}_{\vec{\ell}} \, \prod_{j=1}^k \frac{\tr (H^{\ell_j})}{\ell_j}  
\nonumber \\[0.75em]
&= \frac{{\rt}^{(g)}_{\vec{\ell}}}{\prod_{j=1}^k \ell_j} \
(\emm)^k \left[ s_1^{\ell_1} \, s_2^{\ell_2} \cdots s_k^{\ell_k} \right] \, ,
\end{align}
where 
\begin{equation}
(\emm)^k \left[ s_1^{\ell_1} \, s_2^{\ell_2} \cdots s_k^{\ell_k} \right]
= \sum_{i_1 , \cdots , i_k =1}^N \lambda_{i_1}^{\ell_1} \,  \lambda_{i_2}^{\ell_2} \cdots
 \lambda_{i_k}^{\ell_k} \, .
\end{equation}
For each ${1 \les k \les 2 \mathfrak{g} \,}$, let 
\begin{equation}
T_k (H) = - \, \delta_{k,1} \,  \frac{N}{2t} \, \tr (H^2) + 
\sum_{\substack{[C] \in \mathcal{C} \\[0.2em] \beta_0 (\partial C) =k}}
(N/t)^{\chi (C)} \, T_{\scaleto{[C]}{6pt}} (H) \, .
\end{equation}
Let ${\tilde{T}_k (s_1 , \cdots , s_k)}$ be the polynomial in ${{ \{ s_j \}}_{j=1}^{k}}$ satisfying 
\begin{equation} \label{k-point-inter-1}
(\emm)^k \left[ {\tilde{T}_k (s_1 , \cdots , s_k)} \right] = T_k (H) \, .
\end{equation}
Because of the symmetry of $(\emm)^k \,$, we can replace ${\tilde{T}_k (s_1 , \cdots , s_k)}$ with its symmetrization ${T_k (s_1 , \cdots , s_k) \,}$, i.e.
\begin{equation}
{T_k (s_1 , \cdots , s_k) } = \frac{1}{k!} \sum_{\sigma \in \mathfrak{S}_k} 
{\tilde{T}_k (s_{\sigma(1)} , \cdots , s_{\sigma(k)})} \, ,
\end{equation}
in \eqref{k-point-inter-1}. We refer to the symmetric polynomials ${T_k (s_1 , \cdots , s_k) \,}$, 
${1 \les k \les 2 \mathfrak{g} \,}$, as the \emph{$k$-point interactions}. We will need the following technical lemma on derivatives of ${T_k}$ in the proof of SDEs:

\begin{lemma} \label{derivative-lemma}
For each fixed element $\hat{\lambda}$ of ${E= { \{ \lambda_i \}}_{i=1}^{N} \,}$, we have
\begin{equation}
{\partial_{\hat{\lambda}}} \left( {(\emm)^k \left[ {T_k (s_1 , s_2 , \cdots , s_k)} \right]} \right)
= k \, (\emm)^{k-1} \left[ {T_k^{(1)} (\hat{\lambda} , s_2 , \cdots , s_k)} \right] \, ,
\end{equation}
where
\begin{equation} \label{derivative-lemma-2}
T_k^{(1)} (s_1 , s_2 , \cdots , s_k) \coloneqq {\partial_{s_1}} \, {T_k (s_1 , s_2 , \cdots , s_k)} \, .
\end{equation}
\end{lemma}
\begin{proof}
Let ${\hat{E} = E \backslash \{ \hat{\lambda} \} \,}$. Let
\begin{equation}
\hat{\Lambda}_1 = \left\{ { ( \hat{\lambda} , a_2 , \cdots , a_k ) \, \big| \,
a_l \in \hat{E} \, , \  \forall l } \right\}  ,
\end{equation}
and, for each ${2 \les r \les k \,}$, let
\begin{equation}
\hat{\Lambda}_r =
\bigcup_{2 \les m_2 < \cdots < m_r \les k} \left\{ { ( \hat{\lambda} , a_2 , \cdots , a_k ) \, \big| \,
a_{m_i} =  \hat{\lambda} \, , \ \forall i \, , \quad  \text{and} \quad   a_l \in \hat{E} \, , \ l \neq m_i
} \right\}  .
\end{equation}
We denote ${{(\emm)^k \left[ {T_k (s_1 ,  \cdots , s_k)} \right]} =
\sum_{i_1 , \cdots , i_k =1}^N T_k ( \lambda_{i_1} , \cdots , \lambda_{i_k} ) }$ by 
\begin{equation}
\sum_{\lambda_I \in E^k} T_k (\lambda_I) \, .
\end{equation}
Since ${T_k}$ is a symmetric polynomial, we have
\begin{equation}
\sum_{\lambda_I \in E^k} T_k (\lambda_I) = \sum_{\lambda_I \in {\hat{E}}^k} T_k (\lambda_I)
+ \sum_{r=1}^k \frac{k}{r} \sum_{\lambda_I \in \hat{\Lambda}_r} T_k (\lambda_I) \, ,
\end{equation}
and
\begin{equation}
{\partial_{\hat{\lambda}}} \left( {T_k (\lambda_I)} \right)
= r \, T_k^{(1)} (\lambda_I) \, , \quad \forall {\lambda_I \in \hat{\Lambda}_r} \, .
\end{equation}
Thus, we get
\begin{align}
{\partial_{\hat{\lambda}}}   { \sum_{\lambda_I \in E^k} T_k (\lambda_I)} 
&= k \, \sum_{r=1}^k \sum_{\lambda_I \in \hat{\Lambda}_r} T_k^{(1)} (\lambda_I)
\nonumber \\[0.2em]
&= k  \sum_{\lambda_J \in E^{k-1}} T_k^{(1)} ( \hat{\lambda} , \lambda_J) \, .
\end{align}
\end{proof}

We rewrite the measure ${\di \rho = e^{- \mathcal{S} (D)} \di D }$ over ${\mathcal{H}_N}$ in the following from
\begin{equation}
\di \rho = e^{- \mathcal{S} (D)} \di D = 
\exp \Big( { \sum_{k=1}^{2 \mathfrak{g}} \frac{1}{k!} \ (\emm)^k \left[ {T_k (s_1 , \cdots , s_k)} \right] } \Big) \di H \, .
\end{equation}
By the \emph{Weyl integration formula}, the measure ${\di \rho }$ induces a measure ${\di \tilde{\varrho} \,}$, given by
\begin{equation} \label{induced-distri-1}
\di \tilde{\varrho}= c_N \, \frac{\text{Vol} ({\mathrm{U}}_N)}{N! \, ({2 \pi})^N } \,  \Delta(\bm{\lambda})^2 \, 
\exp \Big( { \sum_{k=1}^{2 \mathfrak{g}} \frac{1}{k!} \ (\emm)^k \left[ {T_k (s_1 , \cdots , s_k)} \right] } \Big) {\prod_{i=1}^N \di \lambda_i} \, ,
\end{equation}
on the space ${\mathbb{R}^N}$ of eigenvalues of Hermitian matrices. In \eqref{induced-distri-1}, 
${\Delta (\bm{\lambda})}$ denotes the Vandermonde determinant, i.e.
\begin{equation}
{\Delta (\bm{\lambda}) = \prod_{ 1 \les i < j \les N } | \lambda_j - \lambda_i |} \, , 
\end{equation}
and ${c_N = {2}^{\frac{N - N^2}{2}} \,}$.\footnote{In addition, the term ${\text{Vol} ({\mathrm{U}}_N)}$ is the volume of the unitary group ${\mathrm{U}}_N$ with respect to the Riemannian volume form corresponding to the induced metric on ${{\mathrm{U}}_N}$ from the inner product ${\langle A , B \rangle = \tr (A B^\ast)}$ on the ambient vector space ${{\mathrm{M}}_N (\mathbb{C}) \,}$.} 

Let 
\begin{equation} \label{compact-omega}
{\Omega = {\hat{\Gamma}}^N \subset {\mathbb{R}^N} \,} ,
\end{equation}
where ${\hat{\Gamma} \subset \mathbb{R}}$ is a strict $\epsilon$-enlargement of the support $\Gamma$ of the large-$N$ spectral distribution $\mu \,$, i.e. 
${\Gamma \subset {\hat{\Gamma}}^{\mathrm{o}} \,}$, and ${\hat{\Gamma} \backslash \Gamma}$ has small Lebesgue measure. We assume that if we replace 
${\di \tilde{\varrho}}$ with 
\begin{equation}
\di \varrho = \mathds{1}_{\Omega} \ {\di \tilde{\varrho}}
\end{equation} 
in the definition of the partition function and the correlators, then they get modified by terms of exponential decay as ${N \to \infty \,}$.

\begin{theorem}
For any ${x , (x_i)_{i \in I} \in \mathbb{C} \backslash \hat{\Gamma} }$, the rank $n$ Schwinger-Dyson equation for the connected correlators of the model (up to the boundary term) is given by
\begin{align} \label{SDE}
&{W_{n+1} (x,x, x_I) + 
\sum_{J \subseteq I} W_{|J| +1} (x , x_J) \, W_{n-|J|} (x, x_{I \backslash J}) }
\nonumber \\[0.2em]
+ \, &\sum_{i \in I} \oint_{C_{\hat{\Gamma}}} \frac{\di \xi}{2 \pi \ci } \,
\frac{W_{n-1} (\xi , x_{I \backslash \{ i \}})}{({x- \xi}) {(x_i - \xi)^2}} 
\nonumber \\[0.2em]
+ \, &\sum_{k=1}^{2 \mathfrak{g}} \sum_{\substack{K \vdash \llbracket 1,k \rrbracket  \\ 
{J_1 \sqcup \cdots \sqcup J_{[K]} = I} }}
{\oint_{C_{\hat{\Gamma}}} \Big[ \, {\prod_{r=1}^k  \frac{\di \xi_r}{2 \pi \ci }} \, \Big] }
\, \frac{ \partial_{\xi_1} T_k (\xi_1 , \cdots , \xi_k)}{(k-1)! \, (x-\xi_1)}
\, \prod_{i=1}^{[K]} {W}_{| {K_i} | + |J_i|} \left( {\xi_{K_i} , x_{J_i}} \right) 
\nonumber \\
= \, &0 \, ,
\end{align}
where ${I = \{ 1, 2, \cdots , n-1 \} \,}$, ${n \gre 1 \,}$, and ${C_{\hat{\Gamma}}}$ is a closed counter-clockwise-oriented contour in an $\epsilon$-tubular neighborhood of $\hat{\Gamma}$ which encloses $\hat{\Gamma} \,$.
\end{theorem}

\begin{proof}
Let ${{ \{ \tau_j \} }_{j=1}^{n-1}}$ be a sequence of parameters. Consider the random variables 
${{\mathrm{X}}_j \,}$, ${ j=1 , \cdots , n-1 \,}$, given by 
\begin{equation} 
{\mathrm{X}}_j =  \sum_{i=1}^N \frac{1}{x_j - \lambda_i} \  , \quad 
 x_j \in \mathbb{C} \backslash \hat{\Gamma} \, .
\end{equation}
Let ${\mathrm{P}^{\vec{\tau}}}$ be the probability measure on ${\mathbb{R}^N}$ defined by
\begin{equation}
\mathbb{E}_{\mathrm{P}^{\vec{\tau}}} [f] \coloneqq \frac{1}
{\mathbb{E}_\mathrm{P} \left[ { \exp \big( \sum_{j=1}^{n-1} \tau_j \, \mathrm{X}_j \big) } \right]} \ \mathbb{E}_\mathrm{P} \Big[ {f \, { \exp \Big( \sum_{j=1}^{n-1} \tau_j \, \mathrm{X}_j \Big) }} \Big] \, ,
\end{equation}
where ${f \in {C_0 (\mathbb{R}^N)} \,}$, and
\begin{equation}
\di \mathrm{P} = \frac{1}{\varrho [1]} \, \di \varrho \, .
\end{equation}
We denote the joint cumulants (resp. moments) of ${{ \{ {\mathrm{X}}_j \} }_{j=1}^{r}}$ with respect to the probability measure ${\mathrm{P}^{\vec{\tau}}}$ by 
${W_r^{\vec{\tau}} (x_1 , \cdots , x_r)}$ (resp. 
${{\widehat{W}}_r^{\vec{\tau}} (x_1 , \cdots , x_r) \,}$).

Consider the Riemannian manifold ${\mathbb{R}^N}$ with the Euclidean metric, and the Riemannian volume form
\begin{equation}
 \di \bm{\lambda} \coloneqq \prod_{i=1}^N { \di \lambda_i} \, .
\end{equation}
Let $\Omega$ be the compact subset of ${\mathbb{R}^N}$ given by \eqref{compact-omega}. Consider the smooth vector field 
\begin{equation} 
V = \sum_{i=1}^N \frac{1}{x- \lambda_i} \, \vec{e_i}
\ , \quad  x \in \mathbb{C} \backslash \hat{\Gamma} \, ,
\end{equation}
where ${\vec{e_i} \,}$, ${i=1 , \cdots , N \,}$, denote the standard constant unit vector fields on ${\mathbb{R}^N}$. One gets the rank $n$ Schwinger-Dyson equation for the connected correlators 
${W_n (x , x_1 , \cdots , x_{n-1})}$, by considering the invariance of 
\begin{equation}
Z_N^{\vec{\tau}} =  
{\mathbb{E}_\mathrm{P} \Big[ { \exp \Big( \sum_{j=1}^{n-1} \tau_j \, \mathrm{X}_j \Big) } \Big]}
=\int { \exp \Big( {\sum_{j=1}^{n-1} \tau_j \,
\Big( {\sum_{i=1}^N \frac{1}{x_j - \lambda_i}} \Big) } \Big)} \di \mathrm{P} \, , 
\end{equation}
under the flow of the vector field $V$. 

We rewrite ${Z_N^{\vec{\tau}}}$ in the following form
\begin{equation}
Z_N^{\vec{\tau}} = \frac{1}{\varrho [1]} \int_\Omega \Psi (\bm{\lambda}) \, \di \bm{\lambda} \, ,
\end{equation}
where ${\Psi (\bm{\lambda}) = \prod_{m=1}^3 \psi_m (\bm{\lambda}) }$ is given by
\begin{equation} 
{\psi}_1 (\bm{\lambda}) = { \exp \Big( {\sum_{j=1}^{n-1} \tau_j \,
\Big( {\sum_{i=1}^N \frac{1}{x_j - \lambda_i}} \Big) } \Big)} \, ,
\end{equation}
\begin{equation}
\psi_2 (\bm{\lambda}) =
\exp \Big( { \sum_{k=1}^{2 \mathfrak{g}} \frac{1}{k!} \,
\sum_{i_1 , \cdots , i_k =1}^N T_k \left( { \lambda_{i_1} , \cdots , \lambda_{i_k}} \right)
} \Big) \, ,
\end{equation}
and
\begin{equation}  
\psi_3 (\bm{\lambda}) = \prod_{ 1 \les i \neq j \les N } { | \lambda_j - \lambda_i | } \, .
\end{equation}
By \eqref{Integral invariance-1}, the invariance of ${Z_N^{\vec{\tau}}}$ under the flow of $V$ is equivalent to
\begin{equation} \label{Integral invariance-2}
\mathbb{E}_{\mathrm{P}^{\vec{\tau}}} \Big[ {{\sum_{m=1}^3 \frac{d \psi_m (V)}{\psi_m}} + div(V)  } \Big] = 0 \, ,
\end{equation}
up to the boundary term.\footnote{Since we are considering the case in which ${\Gamma \subset {\hat{\Gamma}}^{\mathrm{o}} \,}$, i.e. both edges of $\hat{\Gamma}$ are \emph{soft}, the boundary term is of exponential decay as ${N \to \infty}$ (see \cite{Borot-Guionnet-2013}).}

Using the Cauchy integral formula, the term ${{d{\psi}_1 (V)}/{{\psi}_1}}$ can be expressed in the following way:\footnote{Strictly speaking, we should assume that $x$ and ${x_j}$'s are not in an $\epsilon$-tubular neighborhood of ${\hat{\Gamma}}$, so that the function 
${1/ \left( (x- \xi) (x_j - \xi)^2  \right) }$ is holomorphic on that neighborhood of ${\hat{\Gamma}}$.}
\begin{align} 
\frac{d{\psi}_1 (V)}{{\psi}_1}
&= \sum_{j=1}^{n-1}  \tau_j \, \sum_{i=1}^N 
\frac{1}{ (x-\lambda_i) (x_j - \lambda_i)^2} 
\nonumber \\
&=\sum_{j=1}^{n-1}  \tau_j \oint_{C_{\hat{\Gamma}}} \frac{\di \xi}{2 \pi \ci } \,
\Big( { \frac{1}{ (x- \xi) (x_j - \xi)^2} } \Big)
\Big( {\sum_{i=1}^N \frac{1}{\xi - \lambda_i}} \Big)  \, .
\end{align}
By interchanging the integration on ${\mathbb{R}^N}$ with the contour integral, we get
\begin{equation} \label{expec-psi-1}
\mathbb{E}_{\mathrm{P}^{\vec{\tau}}} \left[ {\frac{d{\psi}_1 (V)}{{\psi}_1}} \right]
= \sum_{j=1}^{n-1} \tau_j \oint_{C_{\hat{\Gamma}}} \frac{\di \xi}{2 \pi \ci } \,
{ \frac{W_1^{\vec{\tau}} (\xi)}{ (x- \xi) (x_j - \xi)^2} }  \, .
\end{equation}
Considering Lemma \ref{derivative-lemma}, we have
\begin{align}
&\quad \Big( { d \sum_{i_1 , \cdots , i_k =1}^N T_k \left( { \lambda_{i_1} , \cdots , \lambda_{i_k}} 
 \right) } \Big) (V) 
\nonumber \\[0.3em]
&= k \, \sum_{i=1}^N 
\frac{1}{x- \lambda_i} \, \sum_{j_1 , \cdots , j_{k-1} =1}^N T_k^{(1)} \left( { \lambda_i , \lambda_{j_1} , \cdots , \lambda_{j_{k-1}} } \right) 
\nonumber \\[0.3em]
&=k 
{\oint_{C_{\hat{\Gamma}}} \Big[ \, {\prod_{r=1}^k  \frac{\di \xi_r}{2 \pi \ci }} \, \Big] }
\, \frac{ \partial_{\xi_1} T_k (\xi_1 , \cdots , \xi_k)}{x-\xi_1}
\, \prod_{r=1}^k   \Big( { \sum_{i=1}^N \frac{1}{\xi_r - \lambda_i}} \Big)    \ ,
\end{align}
where the integration is a $k$-times iterated contour integral along ${C_{\hat{\Gamma}} \,}$. Thus, using \eqref{discon-con-correlator}, we have
\begin{align} \label{expec-psi-2}
\mathbb{E}_{\mathrm{P}^{\vec{\tau}}} \left[ {\frac{d{\psi}_2 (V)}{{\psi}_2}} \right]
&= \sum_{k=1}^{2 \mathfrak{g}} \frac{1}{(k-1)!} \,
{\oint_{C_{\hat{\Gamma}}} \Big[ \, {\prod_{r=1}^k  \frac{\di \xi_r}{2 \pi \ci }} \, \Big] }
\, \frac{ \partial_{\xi_1} T_k (\xi_1 , \cdots , \xi_k)}{x-\xi_1}
\ {{\widehat{W}}_k^{\vec{\tau}} (\xi_1 , \cdots , \xi_k) \,}
\nonumber \\[0.4em]
&= \sum_{k=1}^{2 \mathfrak{g}} \sum_{K \vdash \llbracket 1,k \rrbracket}
{\oint_{C_{\hat{\Gamma}}} \Big[ \, {\prod_{r=1}^k  \frac{\di \xi_r}{2 \pi \ci }} \, \Big] }
\, \frac{ \partial_{\xi_1} T_k (\xi_1 , \cdots , \xi_k)}{(k-1)! \, (x-\xi_1)}
\, \prod_{i=1}^{[K]} {W}_{| {K_i} |}^{\vec{\tau}} \left( {\xi_{K_i}} \right) \, .
\end{align}
By similar steps, we get
\begin{equation} \label{expec-psi-3}
\mathbb{E}_{\mathrm{P}^{\vec{\tau}}} \left[ {\frac{d \psi_3 (V)}{\psi_3}} \right]
=   {W_2^{\vec{\tau}} (x,x) + {\left( {W_1^{\vec{\tau}} (x)} \right)}^2}  + 
{\partial_x {\left( {W_1^{\vec{\tau}} (x)} \right)}}  \, ,
\end{equation}
and 
\begin{equation} \label{expec-divV}
\mathbb{E}_{\mathrm{P}^{\vec{\tau}}} \left[ {div (V)} \right] = - \partial_x \left( W_1^{\vec{\tau}} (x) \right) \, .
\end{equation}

By substituting \eqref{expec-psi-1}, \eqref{expec-psi-2}, \eqref{expec-psi-3}, and \eqref{expec-divV} into \eqref{Integral invariance-2}, we get
\begin{align} \label{Integral invariance-3}
&{W_2^{\vec{\tau}} (x,x) + {\left( {W_1^{\vec{\tau}} (x)} \right)}^2}
\nonumber \\[0.2em]
+\, &\sum_{i \in I} \tau_i \oint_{C_{\hat{\Gamma}}} \frac{\di \xi}{2 \pi \ci } \,
{ \frac{W_1^{\vec{\tau}} (\xi)}{ (x- \xi) (x_i - \xi)^2} }
\nonumber \\[0.2em]
+\, &\sum_{k=1}^{2 \mathfrak{g}} \sum_{K \vdash \llbracket 1,k \rrbracket}
{\oint_{C_{\hat{\Gamma}}} \Big[ \, {\prod_{r=1}^k  \frac{\di \xi_r}{2 \pi \ci }} \, \Big] }
\, \frac{ \partial_{\xi_1} T_k (\xi_1 , \cdots , \xi_k)}{(k-1)! \, (x-\xi_1)}
\, \prod_{i=1}^{[K]} {W}_{| {K_i} |}^{\vec{\tau}} \left( {\xi_{K_i}} \right) 
\nonumber \\
= \, &0 \, ,
\end{align}
where ${I = \{ 1, 2, \cdots , n-1 \} \,}$. Considering the definition of joint cumulant, for given finite subsets ${L_i \subset \mathbb{N} \, ,}$ ${i=1 , \cdots , \ell \,}$, of $\mathbb{N} \,$, we have
\begin{equation}
\partial_{\vec{\tau}} \Big|_{\vec{\tau} =0} \Big( \prod_{i=1}^{\ell} W_{|L_i|}^{\vec{\tau}} (\xi_{L_i}) 
\Big) = \sum_{J_1 \sqcup \cdots \sqcup J_{\ell} = I} \
\prod_{i=1}^{\ell} W_{|L_i| + |J_i|} (\xi_{L_i} , x_{J_i}) \, .
\end{equation}
Therefore, by taking the derivative ${\partial_{\vec{\tau}} \Big|_{\vec{\tau} =0}}$ of each term of \eqref{Integral invariance-3}, we get \eqref{SDE}.

\end{proof}

The system of Schwinger-Dyson equations is not closed in the sense that, for each ${n \gre 1 \,}$, the rank $n$ SDE gives an expression for ${W_n (x_1 , \cdots , x_n)}$ in terms of\\ ${W_r (x_1 , \cdots , x_r)}$'s with ${1 \les r \les \max \, \{n+1 \, , \, n+ 2 \mathfrak{g} -1 \} \,}$. However, we will see that they are ``asymptotically'' closed as ${N \to \infty \,}$, and we can solve them to find the coefficients 
${W_{g,n} (x_1 , \cdots , x_n)}$ of the large $N$ expansion of the correlators.

For each ${1 \les k \les 2 \mathfrak{g} \,}$, and ${h \gre 0 \,}$, let ${T_{h,k} (s_1 , \cdots , s_k)}$ be the symmetric polynomial in ${{ \{ s_j \}}_{j=1}^{k}}$ satisfying
\begin{equation}
(\emm)^k \left[ {T_{h,k} (s_1 , \cdots , s_k)} \right] = 
- \, \delta_{h,0} \, \delta_{k,1} \,  \frac{\tr (H^2)}{2} \,  + 
\sum_{\substack{[C] \in \mathcal{C} \\[0.2em] \beta_0 (\partial C) =k \, , \ g (C) = h}}
T_{\scaleto{[C]}{6pt}} (H) \, ,
\end{equation}
where ${g(C)}$ denotes the genus of a surface $C \,$. The $k$-point interactions ${T_k (s_1 , \cdots , s_k)}$ of the model can be rewritten in the following form
\begin{equation} \label{k-interact-N-exp}
{T_k (s_1 , \cdots , s_k)} = \sum_{h \gre 0}  (N/t)^{2-2h-k} \ {T_{h,k} (s_1 , \cdots , s_k)} \, ,
\end{equation}
where the summation includes finitely many terms. 
Considering \eqref{k-interact-N-exp} and \eqref{topological-expansion-1},
for each ${n \gre 1 \,}$, ${g \gre 0 \,}$,  the rank $n$ Schwinger-Dyson equation to order ${N^{3-2g-n}}$ gives:
\begin{align} \label{rank n SDE-N exp}
&\quad \ W_{g-1, \,n+1} (x,x, x_I) + 
\sum_{J \subseteq I \, , \ 0 \les f \les g} W_{f, \, |J| +1} (x , x_J) \, W_{g-f, \,n-|J|} (x, x_{I \backslash J}) 
\nonumber \\[1em]
&+ \, \sum_{i \in I} \oint_{C_{\hat{\Gamma}}} \frac{\di \xi}{2 \pi \ci } \,
\frac{W_{g, \,n-1} (\xi , x_{I \backslash \{ i \}})}{({x- \xi}) {(x_i - \xi)^2}} 
\nonumber \\
&+ \,  \sum_{\substack{1 \les k \les 2 \mathfrak{g} \\[0.1em] 0 \les h}} \,
\sum_{\substack{K \vdash \llbracket 1,k \rrbracket  \\[0.1em] {J_1 \sqcup \cdots \sqcup J_{[K]} = I} }}
\ \sum_{\substack{0 \les f_1 , \cdots , f_{[K]}  \\[0.1em] h+k-[K] + \sum_{i} f_i =g }}
\oint_{C_{\hat{\Gamma}}} \Bigg\{ \Big[ \, {\prod_{r=1}^k  \frac{\di \xi_r}{2 \pi \ci }} \, \Big]  \times
\nonumber \\[1em]
&\hspace{3.5cm} \times {\Big[ \, \frac{ \partial_{\xi_1} T_{h,k} (\xi_1 , \cdots , \xi_k)}{(k-1)! \, (x-\xi_1)} \, \prod_{i=1}^{[K]} {W}_{f_i , \, | {K_i} | + |J_i|} \left( {\xi_{K_i} , x_{J_i}} \right) \Big]}
\Bigg\}
\nonumber \\
&= \, 0 \, .
\end{align}

Before continuing, we introduce the following notations which are used in this article.
Let $V \subset X$ be an open subset of a Riemann surface $X$. We denote by 
$\mathscr{O}$, $\mathscr{O}^\ast$, $\mathscr{M}$, $\Omega$, and $\mathscr{Q}$, the sheaves on $X$ defined by
\begin{itemize}
\item
${\mathscr{O} (V) =}$ holomorphic functions on $V$
\item
${\mathscr{O}^\ast (V) =}$ multiplicative group of nonzero holomorphic functions on $V$
\item
${\mathscr{M} (V) =}$ meromorphic functions on $V$
\item
${\Omega (V) =}$ holomorphic 1-forms on $V$
\item
${\mathscr{Q} (V) =}$ meromorphic 1-forms on $V$,
\end{itemize}
respectively.


\section{Spectral curve}

We start by analyzing the rank one Schwinger-Dyson equation to leading order in $N \,$, i.e. the 
Equation \eqref{rank n SDE-N exp} for ${n=1}$ and ${g=0 \,}$, given by
\begin{equation} \label{leading-one-SDE-1}
{\left( {W_{0,1} (x)} \right)}^2 + 
\sum_{k=1}^2 
{\oint_{C_{\Gamma}} \Big[ \, {\prod_{r=1}^k  \frac{\di \xi_r}{2 \pi \ci }} \, \Big] }
\, \frac{ \partial_{\xi_1} T_{0,k} (\xi_1 , \cdots , \xi_k)}{x-\xi_1}
\, \prod_{r=1}^k W_{0,1} (\xi_r) = 0 \, .
\end{equation}
Recall that
\begin{equation}
T_{0,1} (\xi) = -2 \, t \mathcal{V} (\xi) = - \frac{1}{2} \, \xi^2 + 
\sum_{  \ell  \in {\mathfrak{L}}_{\scaleto{\mathrm{disk}}{4pt}}} \frac{\rt^{(0)}_\ell}{\ell} \, \xi^{\ell} \, ,
\end{equation}
where ${\mathcal{V} (\xi)}$ is given by \eqref{unst-action-1}, and
\begin{equation} \label{T_0,2}
T_{0,2} (\xi , \eta) = \frac{1}{2} \, 
\sum_{ (\ell_1 , \ell_2) \in {\mathfrak{L}}_{\scaleto{\mathrm{cylinder}}{4pt}}}
\frac{\rt^{(0)}_{\ell_1 , \ell_2}}{\ell_1  \ell_2}
\left( \xi^{\ell_1} \eta^{\ell_2} + \xi^{\ell_2} \eta^{\ell_1} \right) \, .
\end{equation}
We fix a simply-connected open neighborhood ${U \subset \mathbb{C}}$ such that 
${\Gamma \subset U \,}$. Consider the following integral operator, called the \emph{master operator} \cite{Borot-2014},
\begin{equation}
\mathcal{O} f (\xi) = 
{  \frac{1}{2 \pi \ci }} \oint_{C_{\Gamma}}    
R(\xi , \eta) \, f(\eta) \di \eta \, , \quad
{\mathcal{O}: \mathscr{O} (U \backslash \Gamma) \to \mathscr{O}(U)} \, ,
\end{equation}
with the kernel
\begin{equation}
R(\xi , \eta) = \partial_{\xi} \, T_{0,2} (\xi , \eta) \, .
\end{equation}
We rewrite \eqref{leading-one-SDE-1} in the following form
\begin{equation} \label{leading-one-SDE-2}
{\left( {W_{0,1} (x)} \right)}^2 +
{\oint_{C_{\Gamma}}  {  \frac{\di \xi}{2 \pi \ci }}  } \, 
\frac{Q(\xi)}{x- \xi} \, W_{0,1} (\xi) = 0 \, ,
\end{equation}
where
\begin{equation} 
Q(\xi) = \left( {\partial_{\xi} \, T_{0,1} (\xi)} \right) + \mathcal{O} W_{0,1} (\xi) \, .
\end{equation}

Since ${W_{0,1} (x) \in \mathscr{O} (\mathbb{C} \backslash \Gamma) \,}$, using \eqref{jump-discon-1} and \eqref{Laurent-exp-W01}, we have
\begin{equation}
{\oint_{C_{\Gamma}}  {  \frac{\di \xi}{2 \pi \ci }}  } \, \xi^{\ell} \, W_{0,1} (\xi) = \mathfrak{m}_\ell
= {\oint_{\xi = \infty}  {  \frac{\di \xi}{2 \pi \ci }}  } \, \xi^{\ell} \, W_{0,1} (\xi)
= Q_{0; \ell} \, ,
\end{equation}
where 
\begin{equation}
\mathfrak{m}_\ell = \int s^{\ell} \varphi (s) \di s \, , \quad \ell \gre 1 \, ,
\end{equation}
denotes the moments of the large-$N$ spectral distribution ${\mu = \varphi (s) \di s \,}$, and ${Q_{0; \ell}}$ is given by \eqref{disk-stuffed-map}. Thus, the polynomial ${Q (\xi)}$ can be expressed in terms of ${\mathfrak{m}_\ell}$'s in the following way:
\begin{align} \label{one-SDE-Qx}
Q(\xi) = - \xi &+ \sum_{  \ell  \in {\mathfrak{L}}_{\scaleto{\mathrm{disk}}{4pt}}} {\rt^{(0)}_\ell} \, \xi^{\ell -1} 
\nonumber \\[0.3em]
&+ \frac{1}{2} \, \sum_{ (\ell_1 , \ell_2) \in {\mathfrak{L}}_{\scaleto{\mathrm{cylinder}}{4pt}}}
{\rt^{(0)}_{\ell_1 , \ell_2}}
\bigg[ \frac{\mathfrak{m}_{\ell_2}}{\ell_2} \, \xi^{\ell_1 -1}  + 
\frac{\mathfrak{m}_{\ell_1}}{\ell_1} \, \xi^{\ell_2 -1}  \bigg]  \, .
\end{align}
Considering ${W_{0,1} (\xi) = O (1/ \xi)}$ as ${\xi \to \infty \,}$, we rewrite the contour integral in \eqref{leading-one-SDE-2} as:
\begin{align}
{\oint_{C_{\Gamma}}  {  \frac{\di \xi}{2 \pi \ci }}  } \, 
\frac{Q(\xi)}{x- \xi} \, W_{0,1} (\xi) 
&= {\oint_{C_{\Gamma}}  {  \frac{\di \xi}{2 \pi \ci }}  } \, 
\frac{ Q(x) - \left( { Q(x) - Q(\xi) }\right)}{x- \xi} \, W_{0,1} (\xi) 
\nonumber \\[0.3em]
&= Q(x) \, W_{0,1} (x) - 
{\oint_{C_{\Gamma}}  {  \frac{\di \xi}{2 \pi \ci }}  } \
\Delta [Q] (x, \xi) \, W_{0,1} (\xi) \, ,
\end{align}
where
\begin{equation}
\Delta [Q] (x, \xi) = \frac{Q(x) - Q(\xi)}{x - \xi}
\end{equation}
denotes the \emph{noncommutative derivative}, aka ``finite difference quotient'', of ${Q(\xi) \,}$. The polynomial 
\begin{equation}
P(x) = - {\oint_{C_{\Gamma}}  {  \frac{\di \xi}{2 \pi \ci }}  } \
\Delta [Q] (x, \xi) \, W_{0,1} (\xi)
\end{equation}
has the following expression in terms of ${\mathfrak{m}_\ell}$'s:
\begin{align} \label{one-SDE-Px}
P(x) = \, 1 &- \sum_{  \ell  \in {\mathfrak{L}}_{\scaleto{\mathrm{disk}}{4pt}}} {\rt^{(0)}_\ell} \, \sum_{n=0}^{\ell-2} \mathfrak{m}_{\ell -2 -n} \ x^{n}
\nonumber \\[0.3em]
&- \frac{1}{2} \, \sum_{ (\ell_1 , \ell_2) \in {\mathfrak{L}}_{\scaleto{\mathrm{cylinder}}{4pt}}}
{\rt^{(0)}_{\ell_1 , \ell_2}}
\bigg[ \frac{\mathfrak{m}_{\ell_2}}{\ell_2} \, 
{\Big( {\sum_{n =0}^{\ell_1 -2} \mathfrak{m}_{\ell_1  -2 -n} \ x^{n}} \Big)}  +
\nonumber \\[0.3em]
&\hspace{4cm}  + \frac{\mathfrak{m}_{\ell_1}}{\ell_1} \,
{\Big( {\sum_{n =0}^{\ell_2 -2} \mathfrak{m}_{\ell_2  -2 -n} \ x^{n}} \Big)}  \bigg]  \, .
\end{align}
Therefore, we have proved:

\begin{proposition} \label{Theorem-quadratic eq}
For the random matrix geometries of type ${(1,0)}$ with the distribution ${\di \rho = e^{- \mathcal{S} (D)} \di D \,}$, the Stieltjes transform ${W_{0,1} (x)}$ of the large-$N$ spectral distribution ${\mu = \varphi (s) \di s}$ of the model satisfies the following quadratic algebraic equation 
\begin{equation} 
y^2 + Q(x) \, y + P(x) = 0 \, ,
\end{equation}
where the polynomials ${Q(x)}$ and ${P(x)}$ are given by \eqref{one-SDE-Qx} and 
\eqref{one-SDE-Px}, respectively. The coefficients of ${Q(x)}$ and ${P(x)}$ depend on the Boltzmann weights ${{\mathfrak{t}}^{(0)}_{\vec{\ell}} \,}$, and the moments 
${\mathfrak{m}_\ell}$ of the large-$N$ spectral distribution $\mu \,$.
\end{proposition}

Let the interval ${\Gamma = [ \mathfrak{a} ,  \mathfrak{b}] \subset \mathbb{R} \,}$, and the open neighborhood ${U \subset \mathbb{C} \,}$, ${\Gamma \subset U }$, be the same as the above-mentioned ones. We recall the \emph{Joukowski map} 
${x: \mathbb{C}\backslash \{ 0 \} \to \mathbb{C}}$ given by
\begin{equation} 
x(z) = \frac{\mathfrak{a}+\mathfrak{b}}{2} + \frac{\mathfrak{b}-\mathfrak{a}}{4} \, \big( {z+ \frac{1}{z}} \big) \, .
\end{equation}
Denote by $\mathbb{T}$ and $\mathbb{D}$ the unit circle and the open unit disk in $\mathbb{C} \,$, respectively. The preimage ${x^{-1} (U \backslash \Gamma)}$ of ${U \backslash \Gamma}$ under the Joukowski map $x$ has two connected components 
${V^e \subset \mathbb{C} \backslash \overline{\mathbb{D}} }$ and 
${V^i \subset \mathbb{D} \,}$, whose common boundary is the unit circle 
$\mathbb{T}$ (see Figure \ref{fig1}). The exterior neighborhood ${V^e} \,$ is mapped to the interior neighborhood ${V^i}$ under ${\iota : z \mapsto 1/z \,}$.

\begin{figure}[t]
\centering 
\includegraphics[width=.4\textwidth]{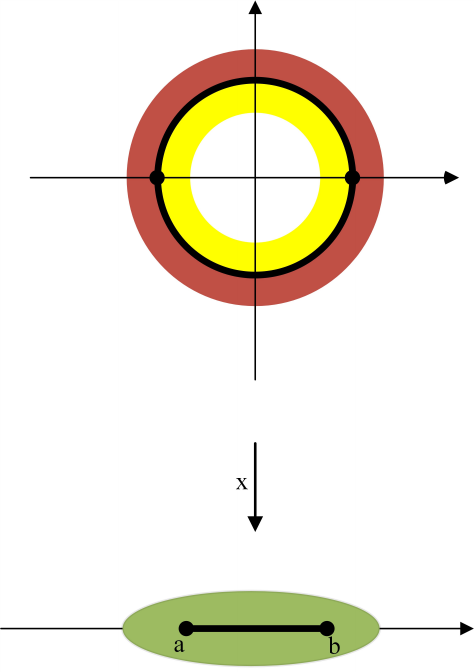}
\caption{Illustration of the Joukowski map in the case where the boundary of $U$ is an ellipse with the foci ${x=\mathfrak{a}}$ and ${x=\mathfrak{b} \,}$. The open neighborhoods $U$, $V^e$, and $V^i$ are colored green, brown and yellow, respectively.}\label{fig1}
\end{figure}

By Proposition \ref{Theorem-quadratic eq} and Lemma \ref{one-cut lemma}, the function 
${W_{0,1} (x) \in \mathscr{O} (U \backslash \Gamma)}$ can be\\ expressed in the following way
\begin{equation} \label{W01-sqrt}
2 \, W_{0,1} (x) = - Q(x) + M(x) \sqrt{(x-\mathfrak{a}) (x-\mathfrak{b})} \, ,
\end{equation}
where ${M(x) \in \mathscr{O}^\ast (U) \,}$, and ${\Gamma = [ \mathfrak{a} , \mathfrak{b} ] \subset \mathbb{R}}$ is the support of the large-$N$ spectral distribution $\mu \,$. Let ${W_{0,1} (x(z))}$ be the pullback of 
${W_{0,1} (x) \in \mathscr{O} (U \backslash \Gamma)}$ under the biholomorphism 
${x |_{V^e} : {V^e} \to U \backslash \Gamma \,}$. Since
\begin{equation}
\sqrt{ {\left( x(z) - \mathfrak{a} \right)} {\left( x(z) - \mathfrak{b} \right)}} = 
\frac{\mathfrak{b}-\mathfrak{a}}{4} \, \big( {z - \frac{1}{z}} \big) \, ,
\end{equation}
considering \eqref{W01-sqrt}, the function ${W_{0,1} (x(z))}$ has an analytic continuation 
to\\ $V = x^{-1} (U) = V^e \cup V^i \cup \mathbb{T} \,$, which is denoted by 
${W_{0,1} (x(z)) \in \mathscr{O} (V) \,}$.

Consider the open neighborhood ${\Upsilon = \mathbb{D} \backslash \overline{V^i}}$ of the point ${z=0 \,}$. The Riemann surface
\begin{equation}
\Sigma = \mathbb{C} \mathbb{P}^1 \backslash \Upsilon \, ,
\end{equation}
which is homeomorphic to a disk, is called the \emph{spectral curve} of the model. Let
\begin{equation}
x: \Sigma \to \mathbb{C} \mathbb{P}^1 \, , \quad x \in \mathscr{M}(\Sigma) \, ,
\end{equation}
be the restriction of the Joukowski map to $\Sigma \,$. The map $x |_V : V \subset \Sigma \to U$ is a two-sheeted ramified covering map with the ramification points 
\begin{equation}
{\mathfrak{R} = \{ z= 1 \, , \, z= -1 \} \,} ,
\end{equation}
and the branch points ${x=\mathfrak{a} \, }$, ${x=\mathfrak{b} \,}$. In addition, the spectral curve $\Sigma$ is equipped with a \emph{local} biholomorphic involution 
\begin{align}
\iota : V &\to V \nonumber \\
z &\mapsto 1/z \ ,
\end{align}
which satisfies ${x \circ \iota = x \,}$.  

Using the Schwinger-Dyson equations \eqref{rank n SDE-N exp} recursively, it can be shown \cite{Eynard-CountingSurfaces} that each ${W_{g,n} (x, x_I)}$ for fixed ${{(x_i)}_{i \in I} \in \mathbb{C} \backslash \Gamma \,}$, initially defined as a holomorphic function on 
${\mathbb{C} \backslash \Gamma \,}$, has a meromorphic continuation ${W_{g,n} (x (z_1), x_I) \in \mathscr{M}(\Sigma) }$ to $\Sigma \,$. By doing the same process for the other arguments ${x_i \,}$, ${i \in I \,}$, of ${W_{g,n} (x(z_1) , x_1 , \cdots , x_{n-1})}$, one gets a meromorphic function 
${W_{g,n} \left( { x(z_1) , \cdots , x(z_n)} \right) }$ on ${\Sigma^n}$.

Let ${K_{\Sigma}}\to {\Sigma}$ be the canonical line bundle, i.e. the holomorphic cotangent bundle, on the spectral curve ${\Sigma}$. Denote by ${\pi_i : {\Sigma}^n \to {\Sigma}}$ the projection map onto the $i$-th component. Let 
\begin{equation}
K_{\Sigma} ^{\boxtimes n} \coloneqq \left({ \pi_1^\ast K_{\Sigma} } \right) \otimes \cdots \otimes 
\left({ \pi_n^\ast K_{\Sigma} } \right)
\end{equation}
be the $n$-times external tensor product of ${K_{\Sigma} \,}$, where ${ \pi_i^\ast K_{\Sigma} }$ denotes the pullback of ${K_{\Sigma}}$ under ${\pi_i \,}$. The sections of the holomorphic line bundle 
${K_{\Sigma} ^{\boxtimes n} \to {\Sigma}^n}$ are referred to as the 
\emph{differentials of degree $n$} over ${{\Sigma}^n \,}$. A differential of degree $n$ is called \emph{symmetric} if it is invariant under the natural action of the symmetric group ${\mathfrak{S}_n}$ on the line bundle ${K_{\Sigma} ^{\boxtimes n} \to {\Sigma}^n \,}$.

In the theory of (blobbed) topological recursion \cite{Eynard-Orantin-2007, Borot-Eynard-Orantin-2015, Borot-2014}, one constructs meromorphic symmetric differentials 
\begin{align} \label{omega-gn-definition}
{\omega_{g,n} (z_1 , \cdots , z_n)} = \, &W_{g,n} \left( { x(z_1) , \cdots , x(z_n)} \right) 
{ dx(z_1) \, dx(z_2) \cdots  dx(z_n )} \nonumber \\[0.2em]
 &+ \, \delta_{g,0} \, \delta_{n,2}  \, \hat{B}_0  (z_1 , z_2 ) 
\end{align}
of degree $n$ from the meromorphic functions ${W_{g,n} \left( { x(z_1) , \cdots , x(z_n)} \right) \,}$, where ${d x(z_i)}$ denotes the pullback ${\pi_i^\ast (dx)}$ of the 1-form ${dx}$ under ${\pi_i : \Sigma^n \to \Sigma \,}$. In \eqref{omega-gn-definition}, the bidifferential, i.e. differential of degree two, 
\begin{equation} \label{B-hat}
\hat{B}_0  (z_1 , z_2 )  = \frac{d x (z_1)\, d x (z_2)}{{[x (z_1) - x (z_2)]}^2}
\end{equation}
is the pullback of the \emph{fundamental symmetric bidifferential of the second kind with biresidue 1} over ${\mathbb{C} \mathbb{P}^1 \times \mathbb{C} \mathbb{P}^1 \,}$,\footnote{For a Torelli marked closed Riemann surface $X$ of genus ${g \gre 0}$, the fundamental symmetric bidifferential $B (p,q)$ of the second kind with biresidue 1 over ${X^2}$ (also referred to as the \emph{Bergman kernel}, in some references) is uniquely characterized by the following conditions (see, e.g., \cite{Fay-ThetaFun-RiemannSurfaces}, \cite{Tyurin-1978} and \cite{Takhtajan-2001} for details):
\begin{itemize}
\item
${B \in H^0 { \big( {X^2 , K_{X} ^{\boxtimes 2} (2 \Delta)} \big) }^{\mathfrak{S}_2}}$, 
and ${\text{Bires}|_{\Delta} \, B =1 \,}$, where ${\Delta \subset X^2 }$ denotes the diagonal divisor.
\item
For each ${p \in X}$, the 1-form ${B(p, \cdot)}$ has vanishing $a_i$-periods, where 
${{\{ a_i , b_i \}}_{i=1}^g}$ is the symplectic basis of ${H_1 (X , \mathbb{Z})}$ specifying the Torelli making of $X$. 
\end{itemize}
} i.e.
\begin{equation} 
B_0 (x_1 , x_2) = \frac{d x_1\, d x_2}{{(x_1 - x_2)}^2} \ ,
\end{equation}
under the map ${(x,x) : \Sigma \times \Sigma \to \mathbb{C} \mathbb{P}^1 \times \mathbb{C} \mathbb{P}^1 \,}$. A couple ${(g,n)}$ is called \emph{stable} if\\ $2-2g-n <0$, i.e. ${(g,n) \neq (0,1) , (0,2) \,}$.

Before continuing, we introduce several operators in the following which are used in our investigation of the differentials ${\omega_{g,n} (z_1 , \cdots , z_n) \,}$. Denote by ${\mathcal{P}_+}$ and ${\mathcal{P}_-}$ the orthogonal idempotents corresponding to the involution ${\iota : V \to V \,}$, given by
\begin{equation} \label{operator-P-1}
\mathcal{P}_+ = \frac{1}{2} (\mathds{1} + \iota^*) \, , \quad \text{and} \quad 
\mathcal{P}_- = \frac{1}{2} (\mathds{1} - \iota^*) \, ,
\end{equation}
respectively, where $\iota^*$ denotes the pullback under $\iota \,$. Let
\begin{equation} \label{operator-P-2}
{\hat{\mathcal{P}}_{\pm}} = 2 \, \mathcal{P}_{\pm} \, .
\end{equation}
The domain of the operators ${\mathcal{P}_{\pm}}$ can be ${\mathscr{M} (V)}$ or 
${\mathscr{Q} (V) \,}$, depending on the context which they are used. For a fixed ${\epsilon>0 \,}$, consider the closed counter-clockwise-oriented contour 
\begin{equation}
{\gamma = \left\{ z \in \mathbb{C} \, \big| \, |z| = 1+ \epsilon \right\} } \subset V^e \, .
\end{equation}
Denote by ${\hat{\mathrm{F}}  \subset \mathscr{Q} (V)}$ the subspace of meromorphic 
1-forms on ${V \subset \Sigma}$ which do not have poles on $\gamma \,$. 
Denote by ${\Omega_\mathrm{inv} (V)}$ (resp. ${\mathscr{O}_{\mathrm{inv}} (V)}$) the subspace of holomorphic 1-forms (resp. functions) on ${V \subset \Sigma}$ which are invariant under the involution $\iota^* \,$. Let ${T_{0,2} (x(z) , x(\zeta)) \,}$, ${z , \zeta \in V \,}$, be the pullback of ${T_{0,2} (\xi , \eta) }$ under the map\\ ${(x,x) : V \times V \to U \times U \,}$. 
Consider the following integral operator
\begin{equation} \label{operator-O-hat}
\hat{\mathcal{O}} \phi (z) = 
{  \frac{1}{2 \pi \ci }} \oint_{\gamma}    
\hat{R} (z, \zeta)  \, \phi(\zeta) \, , \quad
{\hat{\mathcal{O}}: \hat{\mathrm{F}} \to {\Omega_\mathrm{inv} (V)}} \, ,
\end{equation}
with the kernel
\begin{equation}
\hat{R} (z, \zeta) = d_z \, T_{0,2} (x(z) , x(\zeta)) \ \in 
\Gamma \left( { V \times V, \Omega \boxtimes \mathscr{O}} \right) \, ,
\end{equation}
where ${d_z}$ denotes the exterior derivative operator acting on the first argument. The operator 
${\hat{\mathcal{O}}}$ is closely related to the master operator $\mathcal{O}$ in the sense that, for each fixed ${z_I \in \Sigma^{n-1} \,}$, we have
\begin{equation}
\mathcal{O} W_{g,n} (x(z) , x_I) \, dx = {\hat{\mathcal{O}}} \omega_{g,n} (z , z_I) \, ,
\end{equation}
where 
${\mathcal{O}{W_{g,n} (x(z) , \cdot)} \in {\mathscr{O}_{\mathrm{inv}} (V)}}$ denotes the pullback of ${\mathcal{O}{W_{g,n} (x , \cdot)}}$ under ${x: V \to U \,}$.


\section{Fundamental bidifferential \texorpdfstring{${\omega_{0,2} (z , \zeta)}$}{omega 0,2}}

As the next step, we investigate the \emph{large-$N$ spectral covariance} ${W_{0,2} (x_1 , x_2)}$ of the model. Fix ${\zeta \in \Sigma \backslash \overline{V^i} \,}$. Denote by 
${W_{0,2} (x(z) , x(\zeta)) \,}$, ${z  \in \Sigma \,}$, the meromorphic continuation of the pullback of ${W_{0,2} (x_1 , x_2)}$ under the biholomorphism ${x |_{V^e} : {V^e} \to U \backslash \Gamma}$ to 
$\Sigma \,$. By definition of ${W_{0,2} (x_1 , x_2) \,}$, the possible singularities of the function 
${W_{0,2} (x(z) , x(\zeta))}$ can only occur in ${\mathfrak{R} \cup V^i \subset \Sigma \,}$. 
To investigate poles of ${W_{0,2} (x(z) , x(\zeta)) \,}$, we consider the rank two Schwinger-Dyson equation to leading order in $N \,$, i.e. the Equation \eqref{rank n SDE-N exp} for ${(g,n)=(0,2) \,}$, given by 
\begin{align} \label{leading-two-SDE-1}
&\quad \  2 \, W_{0,1} (x_1) W_{0,2} (x_1,x_2)  
+ \oint_{C_{\Gamma}} \frac{\di \xi}{2 \pi \ci } \,
\frac{W_{0,1} (\xi )}{({x_1 - \xi}) {(x_2 - \xi )^2}} 
\nonumber \\[0.5em] 
 &+ \oint_{C_{\Gamma}} \frac{\di \xi}{2 \pi \ci } \,
 \frac{1}{x_1- \xi}
\Big( { \left[ { {\partial_\xi T_{0,1} (\xi)} + \mathcal{O} W_{0,1} ( \xi) } \right] W_{0,2} (\xi , x_2)
+  \left[ \mathcal{O} W_{0,2} (\xi , x_2) \right] W_{0,1} (\xi)} \Big) 
\nonumber \\ 
&= \, 0 \, ,
\end{align}
for ${x_1 \in U \backslash \Gamma \,}$, and fixed ${x_2 \in \mathbb{C} \backslash \Gamma \,}$. We rewrite \eqref{leading-two-SDE-1} in the following form
\begin{align} \label{leading-two-SDE-2}
& -\big[ { 2 \, W_{0,1} (x_1) + {\partial_{x_1} T_{0,1} (x_1)} + \mathcal{O} W_{0,1} ( x_1) } \big] W_{0,2} (x_1 , x_2)
\nonumber \\[0.5em]
= \, &\left[ \mathcal{O} W_{0,2} (x_1 , x_2) \right] W_{0,1} (x_1) 
+ \partial_{x_2} \Big( \frac{W_{0,1} (x_1) - W_{0,1} (x_2)}{x_1 - x_2} \Big)
+ P_{0,2} (x_1 , x_2) \, ,
\end{align}
where ${P_{0,2} (x_1 , \cdot) \in \mathscr{O} (U) \,}$. By \eqref{W01-sqrt}, the function
\begin{equation}
{ 2 \, W_{0,1} (x(z)) + {\partial_{x} T_{0,1} (x(z))} + \mathcal{O} W_{0,1} (x(z)) }
\end{equation}
has a simple zero at the ramification points $\mathfrak{R} \,$. Therefore, considering \eqref{leading-two-SDE-2},\\ ${W_{0,2} (x(z) , x(\zeta))}$ has a simple pole at $\mathfrak{R} \,$.

By considering \eqref{leading-two-SDE-2} as ${x_1 \to s \pm \ci \epsilon \,}$, 
${s \in \Gamma^{\mathrm{o}} \,}$, it can be shown \cite{Borot-2014} that
\begin{equation} \label{leading-two-SDE-jump-1}
\lim_{\epsilon \to 0^+} \big[ W_{0,2} (s+ \ci \epsilon , x_2) 
+ W_{0,2} (s- \ci \epsilon , x_2) \big]
+ \, \mathcal{O} W_{0,2} (s, x_2) + \frac{1}{(s - x_2)^2} = 0 \, , \quad \forall s 
 \in \Gamma^{\mathrm{o}} \, .
\end{equation}
Using the \emph{identity theorem for holomorphic functions}, and the \emph{Riemann's removable singularities theorem}, we deduce from \eqref{leading-two-SDE-jump-1} that 
\begin{equation} \label{monodromy-W_0,2}
\hat{\mathcal{P}}_+  {W_{0,2} (x(z) , x(\zeta))} 
+ \mathcal{O}{W_{0,2} (x(z) , x(\zeta))}  
+ \frac{1}{[ x(z) - x(\zeta)]^2} = 0 
\end{equation}
for all ${z \in V \subset \Sigma \,}$. Considering \eqref{monodromy-W_0,2}, we get the following equation
\begin{equation} \label{monodromy-omega_0,2-1}
\hat{\mathcal{P}}_+  \tilde{\omega}_{0,2} (z , \zeta) 
+ \hat{\mathcal{O}} \tilde{\omega}_{0,2} (z , \zeta)
+ \hat{B}_0  (z , \zeta ) = 0 \, , 
\end{equation}
for the bidifferential
\begin{equation}
\tilde{\omega}_{0,2} (z , \zeta) = {W_{0,2} (x(z) , x(\zeta))} \, dx(z) \, dx (\zeta) \, ,
\end{equation}
where ${\hat{B}_0  (z , \zeta )}$ is given by \eqref{B-hat}.

Consider the function
\begin{equation}
F_\zeta (z) = W_{0,2} (x(z) , x(\zeta)) \, x' (z) \, x'(\zeta)  \, ,
\end{equation} 
where ${x'(z) \equiv {\partial_z \, x(z) \,}}$. We have
\begin{equation}
\tilde{\omega}_{0,2} (z , \zeta) = {F_\zeta (z)} \, dz \, d\zeta \, .
\end{equation}
In the remaining part of this section, we consider the following set of assumptions, and derive an explicit expression for ${F_\zeta (z) \,}$.
\begin{hypothesis} \label{Hypo-1}
\leavevmode
\begin{enumerate}[(i)]
\item
For each fixed ${\zeta \in \mathbb{C} \backslash \overline{\mathbb{D}} \,}$, the function ${F_\zeta (z)}$ has a meromorphic continuation to the whole ${\mathbb{C} \mathbb{P}^1 }$;
\item
The support $\Gamma$ of the large-$N$ spectral distribution $\mu$ of the model is of the form 
${\Gamma = [ -\mathfrak{b} ,\mathfrak{b}] \subset \mathbb{R} \,}$.\footnote{Our approach for analyzing the function ${F_\zeta (z)}$ has a straightforward extension to the general case of 
${\Gamma = [ \mathfrak{a} ,\mathfrak{b}] \subset \mathbb{R} \,}$. However, the computational part becomes more cumbersome.} 
\end{enumerate}
\end{hypothesis}

To get an explicit expression for the meromorphic function ${F_\zeta (z) \,}$, ${z \in {\mathbb{C} \mathbb{P}^1 \,}}$, we find the principal part of the germ of ${F_\zeta (z)}$ at its poles on ${\mathbb{C} \mathbb{P}^1 \,}$, and then analyze the corresponding \emph{Mittag-Leffler problem}. More precisely, we consider the following long exact sequence of \v Cech cohomology groups
\begin{equation}
0 \to H^0 ( {\mathbb{C} \mathbb{P}^1} , \mathscr{O} )
\to H^0 ( {\mathbb{C} \mathbb{P}^1} , \mathscr{M} )
\to H^0 ( {\mathbb{C} \mathbb{P}^1} , \mathscr{M}/ \mathscr{O} )
\to H^1 ( {\mathbb{C} \mathbb{P}^1} , \mathscr{O} )
\to \cdots \ .
\end{equation}
It is well known that
\begin{equation}
H^0 ( {\mathbb{C} \mathbb{P}^1} , \mathscr{O} ) = \mathbb{C} \, , \quad \text{and} \quad 
H^1 ( {\mathbb{C} \mathbb{P}^1} , \mathscr{O} ) = 0 \, .
\end{equation}
Therefore, given a section ${\tilde{f} \in H^0 ( {\mathbb{C} \mathbb{P}^1} , \mathscr{M}/ \mathscr{O} ) \,}$, there exists a meromorphic function 
$f  \in H^0 ( {\mathbb{C} \mathbb{P}^1} , \mathscr{M} ) \,$, unique up to  a constant, whose local singular behavior is given by $\tilde{f} \,$.

The meromorphic function ${F_\zeta (z) }$ does not have any poles in 
${\mathbb{C} \mathbb{P}^1 \backslash \mathbb{D}}$ because the simple zeroes of ${x'(z)}$ at the ramification points $\mathfrak{R}$ cancels out the simple poles of ${W_{0,2} (x(z) , x(\zeta)) \,}$. We use \eqref{monodromy-omega_0,2-1} to analyze the possible poles of ${F_\zeta (z) }$ in 
$\mathbb{D} \,$. Rewrite the function ${T_{0,2} (\xi , \eta) \,}$, given by \eqref{T_0,2}, in the following form
\begin{equation} \label{T_0,2-2}
T_{0,2} (\xi , \eta) = \sum_{k=1}^{d-1} \xi^k \Big( {\sum_{m=1}^{d-k} \nu_{k,m} \, \eta^m} \Big) \, ,
\end{equation}
where each ${\nu_{k,m}}$ is a linear combination of the Boltzmann weights ${\rt^{(0)}_{\ell_1 , \ell_2} \,}$, ${ (\ell_1 , \ell_2) \in {\mathfrak{L}}_{\scaleto{\mathrm{cylinder}}{4pt}} \,}$. We have
\begin{equation} \label{monodromy-omega_0,2-2}
\hat{\mathcal{O}} \tilde{\omega}_{0,2} (z , \zeta)
= {\big( 1 - \frac{1}{z^2} \big)} \sum_{k=1}^{d-1} b_k (\zeta) {\Big( z+ \frac{1}{z} \Big)}^{k-1} 
dz \, d\zeta \, ,
\end{equation}
where, for each ${1 \les k \les {d-1} \,}$, 
\begin{equation}
b_k (\zeta) = \frac{k}{2 \pi \ci} \, \sum_{m=1}^{d-k} \nu_{k,m} (\mathfrak{b}/2)^{k+m}
\oint_\gamma {\Big( \tau+ \frac{1}{\tau} \Big)}^{m} F_\zeta (\tau) \di \tau \, .
\end{equation}
Considering \eqref{monodromy-omega_0,2-1} and \eqref{monodromy-omega_0,2-2}, we get the following equation satisfied by ${F_\zeta (z)}$ on ${\mathbb{C} \backslash \{0\} \,}$:
\begin{equation} \label{monodromy-F-1}
{F_\zeta (z)} = \frac{1}{z^2} \, {F_\zeta (\iota(z))} 
-  \bigg[ {\big( 1 - \frac{1}{z^2} \big)} \sum_{k=1}^{d-1} b_k (\zeta) {\Big( z+ \frac{1}{z} \Big)}^{k-1} \bigg]
- \frac{x'(z) \, x'(\zeta)}{[ x(z) - x(\zeta)]^2} \, .
\end{equation}

We recall the following basic fact from the theory of \emph{projective connections} on Riemann surfaces 
(see, e.g., \cite{Tyurin-1978}). Let ${D \subset \mathbb{C}}$ be an open neighborhood of a point ${p \in \mathbb{C} \,}$. Let $f$ be a biholomorphism on $D \,$. Consider the holomorphic function 
${E_f (v,w)}$ on ${D \times D}$ given by
\begin{equation}
{E_f (v,w)} =
\frac{f'(v) \, f'(w)}{[ f(v) - f(w)]^2} \, , \quad v,w \in D \, .
\end{equation}
Let ${\hat{v} = v -p \,}$, and ${\hat{w} = w -p \,}$. By expanding ${E_f (v,w)}$ as a series in $\hat{v}$ and $\hat{w} \,$, one gets
\begin{equation} \label{Schwarzian}
\frac{f'(v) \, f'(w)}{[ f(v) - f(w)]^2} 
= \frac{1}{(v-w)^2} + \frac{1}{6} \, (Sf) (p) + H(\hat{v} , \hat{w}) \, ,
\end{equation}
as $\hat{v} , \hat{w} \to 0 \,$, where 
\begin{equation}
(Sf) (p) = \frac{f''' (p)}{f'(p)} - \frac{3}{2}  \left( {\frac{f'' (p)}{f'(p)}} \right)^2
\end{equation}
is called the \emph{Schwarzian derivative} of $f \,$, and ${H(\hat{v} , \hat{w})}$ is a sum of terms in $\hat{v}$ and $\hat{w}$ of strictly positive degree.

Considering \eqref{monodromy-F-1}, the function ${F_\zeta (z)}$ has a pole of order two at 
${z= \iota(\zeta) \,}$. Using \eqref{Schwarzian}, we have
\begin{align}
- \, \frac{x'(z) \, x'(\zeta)}{[ x(z) - x(\zeta)]^2} &= - \, \frac{x'(\zeta)}{x' (\iota (\zeta))} \, \cdot \,
\frac{x'(z) \, x'(\iota(\zeta))}{[ x(z) - x(\iota(\zeta))]^2}
\nonumber \\[0.5em]
&= - \, \frac{x'(\zeta)}{x' (\iota (\zeta))}
\Big( \frac{1}{{[ z- \iota(\zeta)]}^2} \,+ \,O (1) \Big) \, ,
\end{align}
as ${z \to \iota(\zeta) \,}$. Thus, the principal part of the germ of ${F_\zeta (z)}$ at ${z= \iota(\zeta) }$ is given by
\begin{equation} \label{function-Q1-z1,z2}
Q_1 (z , \zeta) = - \, \frac{x'(\zeta)}{x' (\iota (\zeta))} \cdot
\frac{1}{{[ z- \iota(\zeta)]}^2}= \frac{1}{{\left( {\zeta z - 1} \right)}^2} \, .
\end{equation}

In addition, by \eqref{monodromy-F-1}, the function ${F_\zeta (z)}$ has a pole at ${z=0 \,}$. Since 
${W_{0,2} (x(z) , x(\zeta))}$ has a zero of order 2 at ${z=\infty \,}$, the term 
${F_\zeta (\iota (z)) / {z^2} \,}$, in \eqref{monodromy-F-1}, is regular at ${z=0 \,}$. Therefore, the principal part of the germ of ${F_\zeta (z)}$ at ${z=0}$ is the same as the principal part of the Laurent polynomial
\begin{equation}
-  \bigg[ {\big( 1 - \frac{1}{z^2} \big)} \sum_{k=1}^{d-1} b_k (\zeta) {\Big( z+ \frac{1}{z} \Big)}^{k-1} \bigg] \, ,
\end{equation}
and is given by
\begin{equation}
Q_2 (z , \zeta) = - \sum_{k=1}^{d-1} \, \frac{b_k (\zeta)}{z^{k+1}} \, \sum_{r=0}^{[\frac{k-1}{2}]} \, 
{\mathrm{c}}_{k,r} \, z^{2r} \, ,
\end{equation}
where
\begin{equation}
{\mathrm{c}}_{k,r} = \frac{2r-k}{k} \, {k \choose r} \, .
\end{equation}

By the Mittag-Leffler theorem, we have
\begin{equation}
F_\zeta (z) = Q_1 (z , \zeta) + Q_2 (z , \zeta) + c(\zeta) \, , \quad c(\zeta) \in H^0 ( {\mathbb{C} \mathbb{P}^1} , \mathscr{O} ) \, .
\end{equation}
Since the function $F_\zeta (z)$ has a zero of order two at ${z= \infty \,}$, the constant function ${c(\zeta)}$ should be equal to zero. To get an explicit expression for ${F_\zeta (z)}$ in terms of ${z , \zeta \,}$, it suffices to find ${b_k (\zeta)}$'s, as a function of $\zeta \,$, explicitly. Note that each ${b_k (\zeta)}$ is a linear combination of the Fourier coefficients
\begin{equation}
a_n (\zeta) = \frac{1}{2 \pi \ci} \oint_\gamma \frac{F_\zeta (z)}{z^{n+1}} \, \di z \, , 
\quad n \in \mathbb{Z} \, ,
\end{equation}
of the restriction of ${F_\zeta (z)}$ to $\gamma \,$. 
Considering the degree of the Laurent polynomial ${Q_2 (z , \zeta)}$, it suffices to find only ${\mathrm{x}_i = a_{-i-1} (\zeta) \,}$, ${1 \les i \les {d-1} \,}$. For each ${1 \les i \les {d-1} \,}$, the identity
\begin{equation}
a_{-i-1} (\zeta) = \frac{1}{2 \pi \ci} \oint_\gamma \big[ {Q_1 (z , \zeta) + Q_2 (z , \zeta)} \big] \, z^i 
\di z
\end{equation}
leads to a linear equation of the following form satisfied by the ${\mathrm{x}_i}$'s:
\begin{equation}
\sum_{j=1}^{d-1} \mathrm{C}_{ij} \, \mathrm{x}_j = \frac{i}{\zeta^{i+1}} \, ,
\end{equation}
where the coefficients ${\mathrm{C}_{ij}}$ depend only on $\mathfrak{b} \,$, and ${\rt^{(0)}_{\ell_1 , \ell_2} \,}$, ${ (\ell_1 , \ell_2) \in {\mathfrak{L}}_{\scaleto{\mathrm{cylinder}}{4pt}} \,}$. 

We assume that the given values to the Boltzmann weights ${\mathfrak{t}^{(0)}_{\vec{\ell}}}$ (or, equivalently, the formal parameters ${\alpha_n \, ,}$ ${3 \les n \les d}$) of the model are such that the above-mentioned matrix ${\left( \mathrm{C}_{ij} \right)}$ is invertible. Hence, we get the explicit expression of 
\begin{equation} \label{function-Fz1,z2-1}
F (z , \zeta) = Q_1 (z , \zeta) + Q_2 (z , \zeta) 
\end{equation}
in terms of ${z, \zeta \,}$, which depends on the support 
${\Gamma = [ - \mathfrak{b}, \mathfrak{b}]}$ of the large-$N$ spectral distribution $\mu \,$, and the Boltzmann weights ${\rt^{(0)}_{\ell_1 , \ell_2} \,}$, ${ (\ell_1 , \ell_2) \in {\mathfrak{L}}_{\scaleto{\mathrm{cylinder}}{4pt}} \,}$. Since the function
\begin{equation} \label{function-Fz1,z2-2}
W_{0,2} (x(z) , x(\zeta)) \, x' (z) \, x'(\zeta)
\end{equation}
is symmetric on its initial domain of definition 
${{\big( {\Sigma \backslash \overline{V^i}} \,  \big)} \times 
{\big( {\Sigma \backslash \overline{V^i}} \,  \big)} \,}$, the function ${F(z , \zeta) \,}$, given by \eqref{function-Fz1,z2-1}, gives, indeed, the presumed meromorphic continuation of 
\eqref{function-Fz1,z2-2} to \\
${\mathbb{C} \mathbb{P}^1 \times \mathbb{C} \mathbb{P}^1 \,}$.\footnote{For the classical formal 1-Hermitian matrix models, in the one-cut regime, the meromorphic continuation of \eqref{function-Fz1,z2-2} to ${\mathbb{C} \mathbb{P}^1 \times \mathbb{C} \mathbb{P}^1}$ is given by ${Q_1 (z , \zeta) \,}$, and it is universal in the sense that it does not depend on the formal parameters of the model (see \cite{Eynard-CountingSurfaces}, Section 3.2.1).} From now on, we consider the restriction of ${F(z , \zeta)}$ to ${\Sigma \times \Sigma \,}$. 

The bidifferential 
\begin{equation}
{\tilde{\omega}}_{0,2} (z ,\zeta) = {F(z , \zeta)} \, dz  \, d \zeta
\end{equation}
over ${\Sigma \times \Sigma}$ has only a pole of order 2 at
\begin{equation}
\tilde{\Delta} = \big\{ { (p, \iota (p)) \, \big| \, p \in V \subset \Sigma} \big\} \, \subset 
\Sigma \times \Sigma \, .
\end{equation}
Considering \eqref{function-Q1-z1,z2}, the singular behavior of ${{\tilde{\omega}}_{0,2} (z ,\zeta)}$ at 
$\tilde{\Delta}$ is as follows:
\begin{equation}
{\tilde{\omega}}_{0,2} (z ,\zeta) = 
- \, \frac{x'(\zeta)}{x' (\iota (\zeta))}
\Big( \frac{1}{{[ z- \iota(\zeta)]}^2} \,+ \,O (1) \Big) \, dz  \, d \zeta \, , 
\quad \text{as} \ z \to \iota (\zeta) \, .
\end{equation}
In addition, the polar divisor of the bidifferential 
\begin{equation} 
\hat{B}_0  (z , \zeta )  = \frac{d x (z)\, d x (\zeta)}{{[x (z) - x (\zeta)]}^2}
\end{equation}
over ${\Sigma \times \Sigma}$ is given by ${ 2 \, {\Delta} + 2 \, \tilde{\Delta} \,}$, where
\begin{equation}
{\Delta} = \big\{ { (p, p) \, \big| \, p \in \Sigma} \big\} \, \subset 
\Sigma \times \Sigma 
\end{equation}
denotes the diagonal divisor. The singular behavior of ${\hat{B}_0  (z , \zeta )}$ at $\Delta$ and 
$\tilde{\Delta}$ is given by
\begin{equation}
\hat{B}_0  (z , \zeta ) = \Big( \frac{1}{{[ z- \zeta]}^2} \,+ \,O (1) \Big) \, dz  \, d \zeta \, , 
\quad \text{as} \ z \to \zeta \, ,
\end{equation}
and
\begin{equation}
\hat{B}_0  (z , \zeta ) =  \frac{\iota^* (d \zeta)}{d \zeta} \, \Big( \frac{1}{{[ z- \iota (\zeta)]}^2} \,+ \,O (1) \Big) \, dz  \, d \zeta \, , 
\quad \text{as} \ z \to \iota (\zeta) \, ,
\end{equation}
respectively. Therefore, the symmetric bidifferential 
\begin{equation} \label{bidiff-omega-0,2-1}
\omega_{0,2} (z , \zeta ) = {\tilde{\omega}}_{0,2} (z ,\zeta) + \hat{B}_0  (z , \zeta )
\end{equation}
has only a pole of order 2 at $\Delta \,$, and its singular behavior is given by
\begin{equation} \label{omega_0,2-near diagonal}
\omega_{0,2} (z , \zeta ) = \Big( \frac{1}{{[ z- \zeta]}^2} \,+ \,O (1) \Big) \, dz  \, d \zeta \, , 
\quad \text{as} \ z \to \zeta \, .
\end{equation}
Succinctly, 
\begin{equation} \label{bidiff-omega-0,2-2}
{\omega_{0,2} (z , \zeta ) \in H^0 { \left( {\Sigma^2 , K_{\Sigma} ^{\boxtimes 2} (2 \Delta)} \right) }^{\mathfrak{S}_2}} \, , 
\end{equation}
and
\begin{equation} \label{bidiff-omega-0,2-3}
{\text{Bires}|_{\Delta} \, \omega_{0,2} (z , \zeta ) =1 \,}.
\end{equation}


\section{Local Cauchy kernel}

Let $\tilde{V}$ be the simply-connected domain which one gets by cutting the neighborhood 
${V \subset \Sigma}$ along a radial line 
${\{ z= r e^{\ci \theta} \, | \, r \in \mathbb{R}_+ \, , \, \text{fixed} \, \theta \} \,}$. Fix ${p_0 \in \tilde{V} \,}$. Considering \eqref{bidiff-omega-0,2-2} and \eqref{bidiff-omega-0,2-3}, one gets a 
\emph{local Cauchy kernel}
\begin{equation} \label{Cauchy-kernel-1}
G(z, \zeta) = \int_{p_0}^{\zeta} \omega_{0,2} (z , \tau )
\end{equation}
by integrating the 1-form ${\omega_{0,2} (\cdot , \tau )}$ on $\tilde{V} $ \cite{Borot-Eynard-Orantin-2015}. We have
\begin{equation}
{G(z, \zeta)} \in \Gamma \left( { \Sigma \times V \, , {\Omega \boxtimes {\mathscr{O} (\Delta)}}} \right)
\, ,
\end{equation}
and 
\begin{equation} \label{Cauchy-kernel-2}
{G(z, \zeta)} =  \Big( \frac{1}{z- \zeta} \,+ \,O (1) \Big) \, dz  \, , 
\quad \text{as} \ z \to \zeta \, .
\end{equation}

Denote by ${{[\phi]}_p \in \mathscr{Q}_p /  \Omega_p}$ the image of the germ of a meromorphic 1-form $\phi \in \mathscr{Q} (\Sigma)$ at a point ${p \in \Sigma}$ under the projection map 
${\pi_p :  \mathscr{Q}_p \to  \mathscr{Q}_p /  \Omega_p} \,$. Denote the set of poles of a meromorphic 1-form ${\phi \in \mathscr{Q} (\Sigma)}$ by ${\mathfrak{P} (\phi)} \,$. Let ${ \tilde{{\mathrm{F}}} \subset \mathscr{Q} (V)}$ be the subspace of meromorphic 1-forms which have finitely many poles on ${V \subset \Sigma \,}$. Consider the operator 
${\mathcal{K} :  \tilde{{\mathrm{F}}} \to \mathscr{Q} (\Sigma)}$ given by
\begin{equation} \label{operator-K-1}
\mathcal{K} \phi (z) =
\sum_{p \in V} \underset{\zeta=p}{\res} \ G(z, \zeta) \, \phi (\zeta)
= \frac{1}{2 \pi \ci} \, \sum_{p \in V} \, \oint_{| \zeta-p | = \epsilon} G(z, \zeta) \, \phi (\zeta) \, .
\end{equation}
Note that, for each ${\phi \in \tilde{{\mathrm{F}}} \,}$, a point ${p \in V}$ contributes to the summation in \eqref{operator-K-1} iff ${p \in \mathfrak{P}(\phi) \,}$. In addition, by \eqref{Cauchy-kernel-2}, we have
\begin{equation} \label{operator-K-2}
{[\mathcal{K} \phi ]}_p = {[ \phi ]}_p \, , \quad \forall p \in \mathfrak{P}( \mathcal{K}\phi) = \mathfrak{P}(\phi) \, .
\end{equation}
We use the same notation $\mathcal{K}$ to denote the operator 
${r_{\Sigma , V} \, \mathcal{K} :  \tilde{{\mathrm{F}}} \to \tilde{{\mathrm{F}}} \,}$, where 
\begin{equation}
{r_{\Sigma , V} : \mathscr{Q} (\Sigma) \to \mathscr{Q} (V)}
\end{equation}
is the restriction map. By \eqref{operator-K-2}, we have 
\begin{equation}
{(\mathds{1} - \mathcal{K}) (\tilde{\mathrm{F}}) \subset \Omega (V) \,}.
\end{equation}
Therefore, ${\mathcal{K} :  \tilde{{\mathrm{F}}} \to \tilde{{\mathrm{F}}}}$ is an idempotent operator, and we have the decomposition
\begin{equation}
\tilde{\mathrm{F}} = \mathcal{K} (\tilde{\mathrm{F}}) \oplus (\mathds{1} - \mathcal{K}) (\tilde{\mathrm{F}}) \, .
\end{equation}

Consider the operator ${\tilde{\mathcal{K}} :  \tilde{{\mathrm{F}}} \to \tilde{{\mathrm{F}}}}$ given by
\begin{equation}
\tilde{\mathcal{K}} \phi (z) =
\sum_{p \in V} \underset{\zeta=p}{\res} \ G(z, \iota(\zeta)) \, \phi (\zeta) \, .
\end{equation}
The following simple lemma will be used several times later.
\begin{lemma} \label{decomposition-lemma}
The operator ${\mathcal{K} :  \tilde{{\mathrm{F}}} \to \tilde{{\mathrm{F}}}}$ can be decomposed as
\begin{equation}
\mathcal{K} = \frac{1}{2} \, {(\mathcal{K} + \tilde{\mathcal{K}}) \,\mathcal{P}_+} 
\, + \, \frac{1}{2} \, {(\mathcal{K} - \tilde{\mathcal{K}}) \, \mathcal{P}_-} \, ,
\end{equation}
where the orthogonal idempotents 
${\mathcal{P}_{\pm} : {\tilde{{\mathrm{F}}}} \to {\tilde{{\mathrm{F}}}}_{\pm}}$ are given by \eqref{operator-P-1}.
\end{lemma}
\begin{proof}
We rewrite the operator $\mathcal{K}$ in the following form
\begin{equation}
\mathcal{K} = \frac{1}{2} \left[ { {(\mathcal{K} + \tilde{\mathcal{K}}) } 
+ {(\mathcal{K} - \tilde{\mathcal{K}}) }} \right] 
\left[ {{\mathcal{P}_+} + {\mathcal{P}_-}} \right] \, .
\end{equation}
Since ${\tilde{\mathcal{K}} = \mathcal{K} \, \iota^* \,}$, we have
\begin{equation}
{(\mathcal{K} \pm \tilde{\mathcal{K}}) } {{\mathcal{P}}_{\mp}} = 0 \, .
\end{equation}
\end{proof}

In the remaining part of this section, we try to explain some elementary aspects of the blobbed topological recursion formula, which will be discussed in the next section, in a simpler setup. Consider the operator 
\begin{equation} \label{operator-T-1}
\mathcal{T} = \hat{\mathcal{P}}_+ + \hat{\mathcal{O}}  \, , \quad 
{\mathcal{T} : \hat{\mathrm{F}} \to  \mathscr{Q}(V)} \, ,
\end{equation}
where the operators ${\hat{\mathcal{P}}_+}$ and $\hat{\mathcal{O}}$ are given by 
\eqref{operator-P-2} and \eqref{operator-O-hat}, respectively. Let ${P \subset V}$ be a fixed finite set of points in ${V \subset \Sigma}$ such that ${P \cap \gamma = \emptyset \,}$. Suppose we are given a set of germs ${s_i \in {\mathscr{Q}_{p_i}}/{\Omega_{p_i}} \,}$, at the points ${p_i \in P \,}$, ${i=1 , \cdots , |P| \,}$, and an $\iota$-invariant holomorphic 1-form ${\psi \in {\Omega}_\mathrm{inv} (V) \,}$. Denote by ${\mathrm{E}\subset H^0 (\Sigma , K_{\Sigma})}$ the subspace of holomorphic 1-forms 
${\tilde{\eta}}$ on $\Sigma$ whose restriction to ${V \subset \Sigma}$ satisfies the homogeneous equation
\begin{equation} \label{operator-T-2}
\mathcal{T} {\tilde{\eta}} = 0 \, , \quad {\tilde{\eta}} \in \mathrm{E} \, .
\end{equation}
Let ${\mathrm{A} \subset \mathscr{Q} (\Sigma)}$ be the set of meromorphic 1-forms $\phi$ on $\Sigma$ whose restriction to ${V \subset \Sigma}$ satisfies the inhomogeneous equation
\begin{equation} \label{operator-T-3}
\mathcal{T} \phi = \psi \, , \quad \phi \in \mathrm{A} \, , 
\end{equation}
and their singularities satisfy
\begin{equation} \label{operator-T-4}
\mathfrak{P} (\phi) = P \, ,  
\end{equation}
and
\begin{equation} \label{operator-T-5}
{[ \phi ]}_{p_i} = s_i \in {\mathscr{Q}_{p_i}}/{\Omega_{p_i}} \, , \quad \forall p_i \in P \, .
\end{equation}
The set ${\mathrm{A} \subset \mathscr{Q} (\Sigma)}$ is an affine space over the vector space $\mathrm{E} \,$. In the following, we investigate the space $\mathrm{A} \,$. 

The operator $\mathcal{K}$ maps the space $\mathrm{A}$ to the meromorphic 1-form 
${\phi_0 \in \mathscr{Q} (\Sigma) \,}$, given by
\begin{equation} \label{phi-0-form-3}
\phi_0 (z) = \sum_{p_i \in P} \underset{\zeta=p_i}{\res} \ G(z, \zeta) \, \tilde{s}_i (\zeta)  \, ,
\end{equation}
where each ${\tilde{s}_i \in \mathscr{Q}_{p_i}}$ is a representative of 
${s_i \in {\mathscr{Q}_{p_i}}/{\Omega_{p_i}}  \,}$, ${\forall p_i \in P \,}$. Using 
\eqref{operator-T-3}, we get 
\begin{equation} \label{phi-0-form-2}
\mathcal{P}_+ \, \phi  \in {\Omega}_\mathrm{inv} (V) \, , \quad \forall \phi \in \mathrm{A} \, . 
\end{equation}
Thus, by Lemma \ref{decomposition-lemma}, we have
\begin{equation} \label{phi-0-form-1}
\phi_0  =  \mathcal{K} \phi  = \frac{1}{2} \, {(\mathcal{K} - \tilde{\mathcal{K}}) \, \mathcal{P}_-} \, \phi \, , 
\quad \forall \phi \in \mathrm{A} \, .
\end{equation}
\begin{proposition} \label{prop-T-K}
The restriction of the 1-form ${\phi_0 \in \mathscr{Q} (\Sigma)}$ to ${V \subset \Sigma}$ satisfies
\begin{equation}
\mathcal{T} \phi_0 = 0 \, .
\end{equation}
\end{proposition}
\begin{proof}
Considering \eqref{monodromy-omega_0,2-1} and \eqref{bidiff-omega-0,2-1}, since 
${{\Omega}_\mathrm{inv} (V) \subset \text{ker}(\hat{\mathcal{O}}) \,}$, we have
\begin{equation} 
\mathcal{T}  {\omega}_{0,2} (z , \zeta) 
=  \hat{B}_0  (z , \zeta ) \, ,
\end{equation}
for each fixed ${\zeta \in \tilde{V}}$. Thus, 
\begin{equation} \label{prop-T-K-1}
\mathcal{T} G(z, \zeta) = \hat{G}_0 (z, \zeta) \, ,
\end{equation}
where 
\begin{equation}
\hat{G}_0 (z, \zeta) = \frac{d x(z)}{x(z) - x(\zeta)}
\end{equation}
is the pullback of the canonical Cauchy kernel $\frac{dx}{x-y}$ over $\mathbb{C}\mathbb{P}^1$ under the map 
\begin{equation}
{(x,x) : \Sigma \times \Sigma \to \mathbb{C} \mathbb{P}^1 \times \mathbb{C} \mathbb{P}^1 \,} .
\end{equation}
Consider the operator ${\mathcal{K}_0 :  \tilde{{\mathrm{F}}} \to \mathscr{Q} (\Sigma)}$ given by
\begin{equation} 
\mathcal{K}_0 \, \varphi (z) =
\sum_{p \in V} \underset{\zeta=p}{\res} \ \hat{G}_0 (z, \zeta) \, \varphi (\zeta) \, .
\end{equation}
By \eqref{prop-T-K-1}, We have
\begin{equation} \label{prop-T-K-2}
\mathcal{T} \mathcal{K} = \mathcal{K}_0
\end{equation}
on the subspace ${{\mathrm{F}} = \tilde{{\mathrm{F}}} \cap \hat{{\mathrm{F}}} \subset  \mathscr{Q} (V)}$ of meromorphic 1-forms which have finitely many poles on ${V \subset \Sigma \,}$, and their poles are not on the contour $\gamma \,$. 

Let ${\tilde{\mathrm{A}} = r_{\Sigma , V} (\mathrm{A}) \,}$. Considering \eqref{phi-0-form-1} 
and \eqref{prop-T-K-2}, it suffices to show that
\begin{equation} \label{prop-T-K-3}
\tilde{\mathrm{A}} \subset \text{ker}(\mathcal{K}_0) \, .
\end{equation}
By decomposing the operator ${\mathcal{K}_0}$ in the same way as in Lemma \ref{decomposition-lemma}, using \eqref{phi-0-form-2}, we get
\begin{equation}
\mathcal{K}_0 \, \phi = 
\frac{1}{2} \, {(\mathcal{K}_0 - \tilde{\mathcal{K}}_0) \, \mathcal{P}_-} \, \phi \, ,
\quad  \forall \phi \in \tilde{\mathrm{A}} \, . 
\end{equation}
On the other hand, we have ${\mathcal{K}_0 = \tilde{\mathcal{K}}_0}$ because 
${\hat{G}_0 (z, \zeta) = \hat{G}_0 (z, \iota (\zeta)) \,}$. Thus, we get \eqref{prop-T-K-3}.

\end{proof}

Let 
\begin{equation}
\mathrm{A}_h = (\mathds{1} - \mathcal{K}) (\mathrm{A}) \, \subset H^0 (\Sigma , K_{\Sigma})
\end{equation}
be the image of $\mathrm{A}$ under the operator ${\mathds{1} - \mathcal{K} \,}$. Considering Proposition \ref{prop-T-K}, the space ${\mathrm{A}_h}$ consists of holomorphic 1-forms 
${{\eta}}$ on $\Sigma$ whose restriction to ${V \subset \Sigma}$ satisfies 
\begin{equation} 
\mathcal{T} {\eta} = \psi \, , \quad {\eta} \in \mathrm{A}_h \, . 
\end{equation}
Therefore, ${\mathrm{A}_h}$ is an affine subspace of ${H^0 (\Sigma , K_{\Sigma})}$ over the vector space $\mathrm{E} \,$.

\begin{lemma} \label{lemma-form-eta-0}
The holomorphic 1-form 
\begin{equation} \label{lemma-form-eta-0-1}
\eta_0 (z) = - \frac{1}{2 \pi \ci} \oint_{\partial \Sigma} G(z, \zeta) \, \psi(\zeta)
\end{equation}
is an element of the affine subspace ${\mathrm{A}_h \subset {H^0 (\Sigma , K_{\Sigma})}\,}$,\footnote{We assume that ${G(\cdot , \zeta)}$ and ${\psi (\zeta)}$ have a continuous extension to 
${\overline{V} \subset \Sigma \,}$.} where the boundary ${\partial \Sigma}$ of the spectral curve $\Sigma$ is assumed to have the induced orientation from $\Sigma$, i.e. it is clock-wise oriented.
\end{lemma}
\begin{proof}
By \eqref{prop-T-K-1}, the restriction of ${\eta_0}$ to ${V \subset \Sigma}$ satisfies
\begin{equation}
\mathcal{T}\eta_0 (z) = - \frac{1}{2 \pi \ci} \oint_{\partial \Sigma} {\hat{G}_0}(z, \zeta) \, \psi(\zeta) \, .
\end{equation}
In addition, we have
\begin{align}
\frac{2}{2  \pi \ci} \oint_{\partial \Sigma} {\hat{G}_0}(z, \zeta) \, \psi(\zeta)
&=\frac{1}{2  \pi \ci} \oint_{\partial \Sigma} {\hat{G}_0}(z, \zeta) \, \psi(\zeta)
+ \frac{1}{2  \pi \ci}  \oint_{\iota (\partial \Sigma)} {\hat{G}_0}(z, \zeta) \, \psi(\zeta) 
\nonumber \\[1.2em]
&= \underset{ \zeta = z}{\res} \, {\hat{G}_0}(z, \zeta) \, \psi(\zeta)
+ \underset{\zeta = \iota (z)}{\res} \, {\hat{G}_0}(z, \zeta) \, \psi(\zeta)
\nonumber \\[0.5em]
&= - 2 \, \psi (z) 
\end{align}
for each ${z \in V \subset \Sigma \,}$. Therefore, we get
\begin{equation}
\mathcal{T} \eta_0 = \psi \, .
\end{equation}

\end{proof}

\begin{lemma} \label{lemma-kernel-T}
In general, the subspace ${\mathrm{E} \subset {H^0 (\Sigma , K_{\Sigma})}}$ is non-trivial, and its dimension depends on the given values to the Boltzmann weights ${\mathfrak{t}^{(0)}_{\vec{\ell}}}$ (or, equivalently, the formal parameters ${\alpha_n \, ,}$ ${3 \les n \les d}$) of the model. 
\end{lemma}
\begin{proof}
For simplicity, we give a proof in the case where the support $\Gamma$ of the large-$N$ spectral distribution $\mu$ of the model is of the form 
${\Gamma = [ -\mathfrak{b} ,\mathfrak{b}] \subset \mathbb{R} \,}$. Consider a holomorphic 1-form 
${{\tilde{\eta}} = f(z) \, dz \in \mathrm{E} \,}$. We have 
\begin{equation} \label{lemma-kernel-T-1}
f(z) = \sum_{n=2}^{\infty} \frac{a_{-n}}{z^n} \, ,  
\end{equation}
where 
\begin{equation}
a_{-n} = \frac{1}{2 \pi \ci} \oint_\gamma f(z) \, z^{n-1} \, \di z \, , \quad n \gre 2 \, ,
\end{equation}
are the Fourier coefficients of the restriction of $f$ to $\gamma \,$. Since 
\begin{equation}
\mathcal{T} {\tilde{\eta}} = 0 \, , \quad \text{on} \ V \subset \Sigma \, ,
\end{equation}
considering \eqref{T_0,2-2}, we get
\begin{equation} \label{lemma-kernel-T-2}
f(z) - \frac{1}{z^2} \, f (\frac{1}{z}) + {\big( 1 - \frac{1}{z^2} \big)} \sum_{k=1}^{d-1} b_k  {\Big( z+ \frac{1}{z} \Big)}^{k-1} = 0 \, ,
\end{equation}
on ${V \subset \Sigma \,}$, where, for each ${1 \les k \les {d-1} \,}$,
\begin{equation}
b_k  = \frac{k}{2 \pi \ci} \, \sum_{m=1}^{d-k} \nu_{k,m} (\mathfrak{b}/2)^{k+m}
\oint_\gamma {\Big( \tau+ \frac{1}{\tau} \Big)}^{m} f(\tau) \di \tau \, .
\end{equation}
By substituting \eqref{lemma-kernel-T-1} into \eqref{lemma-kernel-T-2}, we get ${d-1}$ linear equations of the form
\begin{equation}
\sum_{j=1}^{d-1} \mathrm{C}_{ij} \, \mathrm{x}_j = 0 \, ,
\end{equation}
satisfied by ${\mathrm{x}_i = a_{-i-1} \,}$, ${1 \les i \les {d-1} \,}$, where the coefficients ${\mathrm{C}_{ij}}$ depend only on $\mathfrak{b} \,$, and ${\rt^{(0)}_{\ell_1 , \ell_2} \,}$, ${ (\ell_1 , \ell_2) \in {\mathfrak{L}}_{\scaleto{\mathrm{cylinder}}{4pt}} \,}$. Thus, the dimension of the subspace ${\mathrm{E} \subset {H^0 (\Sigma , K_{\Sigma})}}$ equals the dimension of the kernel of the matrix ${(\mathrm{C}_{ij}) \, }$.

\end{proof}

Considering Proposition \ref{prop-T-K}, Lemma \ref{lemma-form-eta-0} and Lemma \ref{lemma-kernel-T}, in general, a 1-form ${\phi \in \mathrm{A}}$ can be written as
\begin{equation}
\phi = \phi_0 + \eta_0 + {\tilde{\eta}}
\end{equation}
for some ${{\tilde{\eta}} \in \mathrm{E}\, }$. From now on, we consider the following assumption:
\begin{hypothesis} \label{Hypo-2}
The given values to the Boltzmann weights ${\mathfrak{t}^{(0)}_{\vec{\ell}}}$ (or, equivalently, the formal parameters ${\alpha_n \, ,}$ ${3 \les n \les d}$) of the model are such that the subspace $\mathrm{E} \subset {H^0 (\Sigma , K_{\Sigma})}$ of global holomorohic 1-forms ${\tilde{\eta} \,}$, whose restriction to ${V \subset \Sigma}$ satisfies the homogeneous equation
\begin{equation}
\mathcal{T} {\tilde{\eta}} = 0 \, ,
\end{equation}
is trivial.
\end{hypothesis}
\noindent
Therefore, a 1-form 
$\phi \in \mathscr{Q} (\Sigma) \,$, satisfying 
\eqref{operator-T-3}-\eqref{operator-T-5}, is uniquely given by 
\begin{equation}
\phi = \phi_0 + \eta_0 \, ,
\end{equation}
where the 1-forms ${\phi_0}$ and ${\eta_0}$ are defined by \eqref{phi-0-form-3} and \eqref{lemma-form-eta-0-1}, respectively.


\section{Blobbed topological recursion formula}

In this section we show that all the stable coefficients ${W_{g,n} (x_1 , \cdots , x_n)}$ of the large $N$ expansion of the correlators of our model can be computed recursively using the blobbed topological recursion formula \cite{Borot-2014}. In the following, without further explicit mention, a couple ${(g,n)}$ is assumed to be stable, i.e. ${(g,n) \neq (0,1) , (0,2) \,}$.

Let ${f (x) \in \mathscr{O} (\mathbb{C} \backslash \Gamma)}$ be a holomorphic function on 
${\mathbb{C} \backslash \Gamma}$ with a jump discontinuity on $\Gamma$. We denote by 
${\delta f(s)}$ and ${\sigma f (s)}$ the functions given by
\begin{equation} 
\delta f(s) = \lim_{\epsilon \to 0^+} [f (s + \ci \epsilon) - f (s - \ci \epsilon)] \, , 
\quad s \in {\Gamma^{\mathrm{o}}} \subset \mathbb{R} \, ,
\end{equation}
and
\begin{equation} 
\sigma f (s) = \lim_{\epsilon \to 0^+} [f (s + \ci \epsilon) + f (s - \ci \epsilon)] \, ,
\quad s \in {\Gamma^{\mathrm{o}}} \subset \mathbb{R} \, ,
\end{equation}
respectively. By considering the rank $n$ Schwinger-Dyson equation to order 
${N^{3-2g-n}}$, given by \eqref{rank n SDE-N exp}, as ${x \to s \pm \ci \epsilon \,}$, ${s \in {\Gamma^{\mathrm{o}}}}$, we get
\begin{align} \label{jump-discon-rank-n-SDE}
&\quad \ \delta W_{0,1} (s) \, \big\{ {\sigma W_{g,n} (s, x_I) + \mathcal{O} W_{g,n} (s, x_I)
+ \partial_s V_{g,n} (s; x_I)} \big\} 
\nonumber \\[0.5em]
&+ \, \delta W_{g,n} (s, x_I) \, \big\{ {\sigma W_{0,1} (s) + \mathcal{O} W_{0,1} (s)
+ \partial_s T_{0,1} (s)} \big\}  
\nonumber \\[0.5em]
&+ \, \sum_{i \in I} \delta W_{g, \, n-1} (s, x_{I \backslash \{i \}}) \, \Big\{ {\sigma W_{0,2} (s , x_i) + \mathcal{O} W_{0,2} (s, x_i)
+ \frac{1}{(s-x_i)^2}} \Big\} 
\nonumber \\[0.3em]
&+ \, \sum_{\substack{J \subseteq I \, , \, 0 \leqslant f \leqslant g  \\  (J,f) \neq (\emptyset , 0) , (I,g) , 
\\ \quad \ (I \backslash \{ i \} , g )}}
\delta W_{f, \, {|J| +1}} (s,x_J) \,
\Big\{ \sigma W_{{g-f} ,  \, {n - |J|}} (s, x_{I \backslash J}) \, +
\nonumber \\
&\hspace{4.3cm} + \mathcal{O} W_{{g-f} , \,  {n - |J|}} (s, x_{I \backslash J})
+ { \partial_s V_{{g-f} , \,  {n - |J|}} (s; x_{I \backslash J}) } \Big\} 
\nonumber \\[0.3em]
&+ \, \delta_{s,2} \, \big\{ { \sigma_{s,1} W_{{g-1},  \, {n+1}} (s, s , x_I) 
+ \mathcal{O} W_{{g-1}, \,  {n+1}} (s, s, x_I)
+ \partial_{s} V_{{g-1},  \, {n+1}} (s; s, x_I) } \big\} 
\nonumber \\[0.3em]
&= \, 0 \, ,
\end{align}
where ${\delta_{s,2}}$ (resp. ${\sigma_{s,1}}$) means that the operator $\delta$ (resp. $\sigma$) is acting on the second (resp. first) argument. In \eqref{jump-discon-rank-n-SDE}, the functions ${V_{g,n} (x;x_I) \,}$, called the \emph{potentials for higher topologies} \cite{Borot-2014},\footnote{We consider a slightly modified version of Equation (4.13) in \cite{Borot-2014} to get \eqref{higher potential}.} are given by
\begin{align} \label{higher potential}
V_{g,n} (x ; x_I) = \, &\delta_{g,2} \, \delta_{n,1} \, T_{2,1} (x)
\nonumber \\
&+ \sum_{\substack{2 \les k \les 2 \mathfrak{g}  \\[0.1em]  0 < h}} \
\sum_{\substack{ K  \vdash \llbracket 2,k \rrbracket \\[0.1em]  
0 \les f_1 , \cdots , f _{[K]}  \\[0.1em]
 h+ {(k-1)} - [K] + \sum_i f_i = g \\[0.1em]
J_1 \sqcup \cdots \sqcup J_{[K]} = I }}
\oint_{C_{{\Gamma}}} \Bigg\{  \Big[ \, {\prod_{r=2}^k  \frac{\di \xi_r}{2 \pi \ci }} \, \Big] \times
\nonumber \\[0.5em]
&\hspace{1.5cm} \times \Big[ \, \frac{T_{h,k} (x, \xi_2 , \cdots , \xi_k)}{(k-1)!}
\, \prod_{i=1}^{[K]} W_{f_i, \, {|K_i | + |J_i |}} (\xi_{K_i} , x_{J_i}) \Big] \Bigg\} \, .
\end{align}

Let ${\mathcal{V}_{g,n} (z ; z_I)}$ be the differential of degree ${n-1}$ over ${\Sigma^n}$ given by
\begin{equation}
 \mathcal{V}_{g,n} (z ; z_I) =  V_{g,n} (x(z); x(z_I))  \prod_{i \in I} dx(z_i) \, .
\end{equation}
Consider the differential
\begin{equation}
d_z \mathcal{V}_{g,n} (z ; z_I) = \big[ \partial_{x(z)} V_{g,n} (x(z); x(z_I)) \big]  \, dx(z) \prod_{i \in I} dx(z_i) 
\end{equation}
of degree ${n}$ over ${\Sigma^n \,}$. From now on, we consider a \emph{fixed} ${z_I \in \Sigma^{n-1}\,}$. The restriction of the exact 1-form 
\begin{equation}
d_z \mathcal{V}_{g,n} (z ; z_I) = \big[ \partial_{x(z)} V_{g,n} (x(z); x(z_I)) \big] \, dx(z)
\end{equation}
to ${V \subset \Sigma}$ is in ${{\Omega}_{\mathrm{inv}} (V) \,}$. Using 
\eqref{jump-discon-rank-n-SDE}, by induction on ${2g+n-2 \,}$, it can be shown \cite{Borot-2014} that each 1-form ${\omega_{g,n} (z,z_I)}$ satisfies
\begin{equation}
\mathcal{T} {\omega_{g,n} (z,z_I)} = - {d_z\mathcal{V}_{g,n} (z ; z_I)}
\end{equation}
on ${V \subset \Sigma \,}$.

By recasting the Schwinger-Dyson equations, one can show \cite{Borot-2014} that
\begin{align} \label{quadratic differential-1}
\omega_{{g-1}, \, {n+1}} (z , \iota (z) , z_I) 
\, + \sum_{  J \subseteq I  \, , \  0 \les f \les g }
\omega_{f, \, {|J| +1}} (z , z_J) \, \omega_{{g-f} , \, {n - |J|}} (\iota (z) , z_{I \backslash J} &)
\nonumber \\
= \mathcal{Q}_{g,n} (z ; z_I &) \, ,
\end{align}
where ${\mathcal{Q}_{g,n} (z ; z_I)}$ is a \emph{local} holomorphic quadratic differential, i.e. a local section of ${{K_{\Sigma}^{\otimes 2}} \to \Sigma}$ over ${V \subset \Sigma}$, with double zeros at the ramification points $\mathfrak{R} \,$. One can rewrite \eqref{quadratic differential-1} in the following form
\begin{equation} \label{quadratic differential-2}
\mathcal{P}_- \, \omega_{g,n} (z, z_I) 
= \frac{1}{  \hat{\mathcal{P}}_- \,\omega_{0,1} (z)} \, \big[ \mathcal{E}_{g,n} (z , \iota(z) ; z_I) + \tilde{\mathcal{Q}}_{g,n} (z ; z_I) \big] \, ,
\end{equation}
where 
\begin{align} \label{quadratic differential-3}
\mathcal{E}_{g,n} (z , \iota(z) ; z_I) = \ &\omega_{{g-1}, \, {n+1}} (z , \iota (z) , z_I) 
\nonumber \\
&+ \sum_{ \substack{ J \subseteq I  \, , \  0 \les f \les g \\[0.2em] (J,f) \neq (\emptyset ,0) \, , \, (I ,g)} }
\omega_{f, \, {|J| +1}} (z , z_J) \, \omega_{{g-f} , \, {n - |J|}} (\iota (z) , z_{I \backslash J}) \, ,
\end{align}
and 
\begin{equation} \label{quadratic differential-4}
\tilde{\mathcal{Q}}_{g,n} (z ; z_I) = 
\mathcal{P}_+ \, \omega_{g,n} (z, z_I) \ \hat{\mathcal{P}}_+ \,\omega_{0,1} (z)
- \mathcal{Q}_{g,n} (z ; z_I) \, .
\end{equation}

Considering \eqref{W01-sqrt}, the zeroes of the 1-from ${  \hat{\mathcal{P}}_- \,\omega_{0,1} (z)}$ in 
${V \subset \Sigma}$ occur only at the ramification points $\mathfrak{R} \,$, and their order is exactly two. Since the quadratic differential ${\tilde{\mathcal{Q}}_{g,n} (z ; z_I)}$ has double zeroes at $\mathfrak{R} \,$, the 1-form
\begin{equation}
\frac{1}{  \hat{\mathcal{P}}_- \,\omega_{0,1} (z)} \, \tilde{\mathcal{Q}}_{g,n} (z ; z_I) 
\end{equation}
is holomorphic on $V$. Therefore, the singularities of ${\omega_{g,n} (z, z_I)}$ in $V$ is the same as the 1-form
\begin{equation} 
\frac{1}{  \hat{\mathcal{P}}_- \,\omega_{0,1} (z)} \, \mathcal{E}_{g,n} (z , \iota(z) ; z_I) \, . 
\end{equation}
Using \eqref{quadratic differential-2}, by induction on ${2g+n-2 \,}$, it can be shown that the poles of the 1-form ${\omega_{g,n} (z, z_I)}$ occur only at the ramification points $\mathfrak{R} \,$. 

Considering \eqref{quadratic differential-2} and \eqref{phi-0-form-1}, we can use the local Cauchy kernel 
${G(z , \zeta)}$ to construct a meromorphic 1-form
\begin{equation}
{\tilde{\omega}}_{g,n} (z, z_I) = \frac{1}{2} \, {(\mathcal{K} - \tilde{\mathcal{K}})}
\Big[ \frac{1}{  \hat{\mathcal{P}}_- \,\omega_{0,1} (z)} \, \mathcal{E}_{g,n} (z , \iota(z) ; z_I) \Big]
\end{equation}
such that 
\begin{equation}
\Phi_{g,n} (z ; z_I)  = {\omega}_{g,n} (z, z_I) - {\tilde{\omega}}_{g,n} (z, z_I) 
\end{equation}
is a holomorphic 1-form on $\Sigma \,$, satisfying 
\begin{equation}
\mathcal{T} \Phi_{g,n} (z ; z_I)  = - {d_z\mathcal{V}_{g,n} (z ; z_I)} \, , \quad \text{on} \ V \subset \Sigma \, .
\end{equation}
The 1-form ${{\tilde{\omega}}_{g,n} (z, z_I)}$ can be expressed in terms of the \emph{recursion kernel} (see \cite{Eynard-Orantin-2007}, \cite{Borot-Eynard-Orantin-2015}) 
\begin{equation}
K(z , \zeta) =  \frac{1}{2} \, \frac{ \int_{\iota (\zeta)}^{\zeta}  \omega_{0,2} (z , \tau)}
{{\omega_{0,1} (\zeta) - \omega_{0,1} (\iota (\zeta))}} 
\ \in \Gamma (\Sigma \times V , \Omega \boxtimes \Omega^{-1})
\end{equation}
in the following form
\begin{equation} \label{BTR-1}
{\tilde{\omega}}_{g,n} (z, z_I) = \sum_{p \in \mathfrak{R}} \underset{\zeta=p}{\res} \
K(z , \zeta) \, \mathcal{E}_{g,n} (\zeta , \iota(\zeta) ; z_I) \, .
\end{equation}
Considering Lemma \ref{lemma-form-eta-0} and Hypothesis \ref{Hypo-2}, the 1-form ${\Phi_{g,n} (z ; z_I)}$ is uniquely given by
\begin{align} \label{BTR-2}
\Phi_{g,n} (z ; z_I) 
&= \frac{1}{2 \pi \ci} \oint_{\partial \Sigma} G(z, \zeta) \, {d_{\zeta}\mathcal{V}_{g,n} (\zeta ; z_I)} 
\nonumber \\[0.5em]
&= - \frac{1}{2 \pi \ci} \oint_{\partial \Sigma} \omega_{0,2}(z, \zeta) \, {\mathcal{V}_{g,n} (\zeta ; z_I)} \, .
\end{align}
Note that, for each stable ${(g,n)}$, the right-hand side of \eqref{BTR-1} and \eqref{BTR-2} involves only ${\omega_{g' , n'}}$ with ${2g'+ n' -2 < 2g+n -2 \,}$. Therefore, we have
\begin{theorem}
For the random matrix geometries of type ${(1,0)}$ with the distribution ${\di \rho = e^{- \mathcal{S} (D)} \di D \,}$, all the stable ${\omega_{g,n} \,}$, ${2g+n-2 >0 \,}$, can be computed recursively, using the blobbed topological recursion formula given by
\begin{equation}\label{BTR-3}
{\omega}_{g,n} (z, z_I) 
= - \frac{1}{2 \pi \ci} \oint_{\partial \Sigma} \omega_{0,2}(z, \zeta) \, {\mathcal{V}_{g,n} (\zeta ; z_I)}
+ \sum_{p \in \mathfrak{R}} \underset{\zeta=p}{\mathrm{Res}} \
K(z , \zeta) \, \mathcal{E}_{g,n} (\zeta , \iota(\zeta) ; z_I) \, .
\end{equation}
The initial data for the recursion relation \eqref{BTR-3} is the 1-form ${\omega_{0,1} (z)}$ and the fundamental bidifferential ${\omega_{0,2} (z, \zeta) \,}$.
\end{theorem}


\section*{Acknowledgment}
\addcontentsline{toc}{section}{Acknowledgment}
We would like to thank Ga\"etan Borot for illuminating discussions and sharing his insights  at an early stage of this work. We also thank Matilde Marcolli for bringing \cite{Gesteau-2019} to our attention. We would like to thank the referee for constructive remarks and suggestions for further development of the line of research we initiated in this paper which we hope to come back to them soon.

\addcontentsline{toc}{section}{References}

\end{document}